\documentclass[conference]{IEEEtran}
\usepackage{cite}
\usepackage{graphicx}
\usepackage{psfrag}
\usepackage{subfigure}
\usepackage{url}
\usepackage{amsmath}
\usepackage{array}
\usepackage{amssymb}
\usepackage{amsfonts}
\usepackage{graphicx}
\usepackage{epstopdf}

\newtheorem{lemma}{Lemma}
\newtheorem{observation}{Observation}

\newtheorem{definition}{Definition}
\newtheorem{theorem}{Theorem}
\newtheorem{example}{Example}

\title{Physical Layer Network Coding for Two-Way Relaying with QAM}
\begin{document}

\author{
\authorblockN{Vishnu Namboodiri$^\dagger$, Kiran Venugopal$^\ddagger$ and B. Sundar Rajan$^\ddagger$}\\
\authorblockA{$^\dagger$Qualcomm India Private Limited, Hyderabad, India- 500081$^\ast$  \\
$^\ddagger$Dept. of ECE, Indian Institute of Science \\
Bangalore 560012, India\\
namboodiri.vishnu@gmail.com, \{kiran.v, bsrajan\}@ece.iisc.ernet.in\\
}
%\authorblockN{Kiran Venugopal}\\
%\authorblockA{Dept. of ECE, Indian Institute of Science \\
%Bangalore 560012, India\\
%Email: \\
%}
%\and
%\authorblockN{B. Sundar Rajan}
%\authorblockA{Dept. of ECE, Indian Institute of Science, \\Bangalore 560012, India\\
%Email: 
%}
}

\maketitle
\thispagestyle{empty}	
\let\thefootnote\relax\footnote{$^\ast$The work was done when the first author was with Indian Institute of Science, Bangalore.

 Part of the content of this paper appeared in IEEE Global Telecommunications Conference(GLOBECOM 2012), CA, USA, 3-7 Dec. 2012.
}
\vspace{-0.15 in} 
%%%%%%%%

%%%%%%%%%%%%%%%%%%%%%%%%%%%%%%%%%%%%%%%%%%%%%%%%%%%%%%%%%%%%%%%%%%%%%%%%%%%%%%%%%%%%%
\begin{abstract}
The design of modulation schemes for the physical layer network-coded two way relaying scenario was studied in \cite{ZLL}, \cite{PoY}, \cite{APT1} and \cite{APT2}. In \cite{NVR} it was shown that every network coding map that satisfies the exclusive law is representable by a Latin Square and conversely, and this relationship can be used to get the network coding maps satisfying the exclusive law. But, only the  scenario in which the end nodes use $M$-PSK signal sets is addressed in \cite{NVR} and \cite{VNR}. In this paper, we address the case in which the end nodes use $M$-QAM signal sets. In a fading scenario, for certain channel conditions $\gamma e^{j \theta}$, termed singular fade states, the MA phase performance is greatly reduced. By formulating a procedure for finding the exact number of singular fade states for QAM, we show that square QAM signal sets give lesser number of singular fade states compared to PSK signal sets. This results in superior performance of $M$-QAM over $M$-PSK. It is shown that the criterion for partitioning the complex plane, for the purpose of using a particular network code for a particular fade state, is different from that used for $M$-PSK. Using a modified criterion, we describe a procedure to analytically partition the complex plane representing the channel condition. We show that when $M$-QAM ($M >4$) signal set is used, the conventional XOR network mapping fails to remove the ill effects of $\gamma e^{j \theta}=1$, which is a singular fade state for all signal sets of arbitrary size. We show that a doubly block circulant Latin Square removes this singular fade state for $M$-QAM.
\end{abstract}
%\vspace{.1 in}
%%%%%%%%%%%%%%%%%%%%%%%%%%%%%%%%%%
\section{Preliminaries and Background}
%\vspace{-.1 in}
We consider the two-way wireless relaying scenario shown in Fig.1, where bi-directional data transfer takes place between the nodes A and B with the help of the relay R. It is assumed that all the three nodes operate in half-duplex mode. The relaying protocol consists of the following two phases: the \textit{multiple access} (MA) phase, during which A and B simultaneously transmit to R using identical square $M$-QAM signal sets and the \textit{broadcast} (BC) phase during which R transmits to A and B using either a square $M$-QAM signal set or a constellation of size more than $M$. Network coding is employed at R in such a way that A (B) can decode the message of B (A), given that A (B) knows its own message. 

%%%%%%%%%%%%%%%%%%%%%%%%%%%%%%%%%%
\begin{figure}[htbp]
\label{fig:relay_channel}
\centering
\subfigure[MA Phase]{
\includegraphics[totalheight=.85in,width=2in]{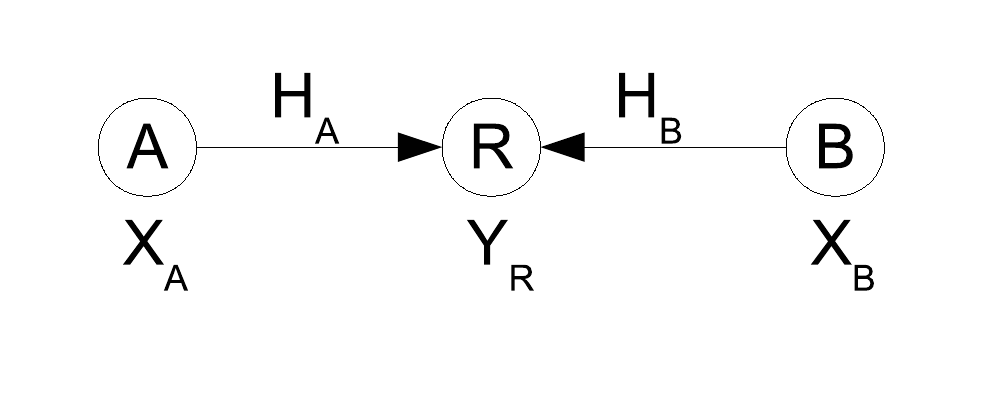}
\label{fig:phase1}
}
\vspace{-0.5 in}
\subfigure[BC Phase]{
\includegraphics[totalheight=.85in,width=2in]{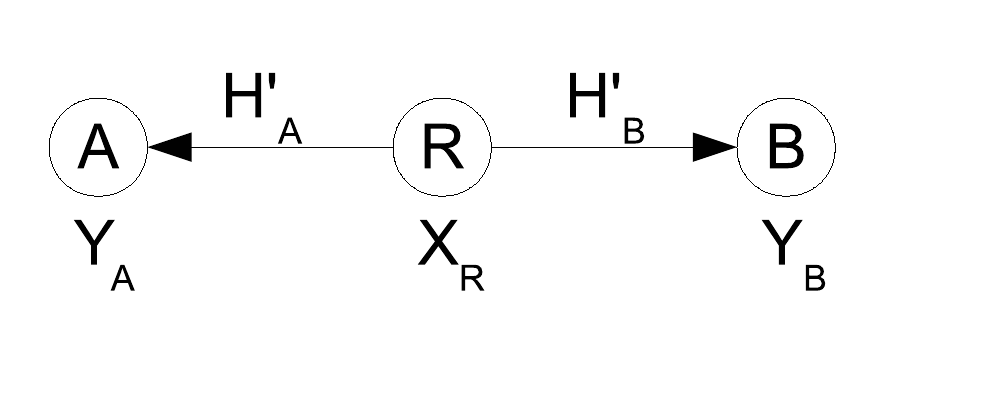}
\label{fig:phase2}
}
\vspace{.4 in}
\caption{The Two Way Relay Channel}
\end{figure}
%\vspace{-0.1 in}
%%%%%%%%%%%%%%%%%%%%%%%%%%%%%%%%
%%%%%%%%%%%%%%%%%%%%%%%%%%%%%% 

\subsection{Background}
 The concept of physical layer network coding has attracted a lot of attention in recent times. The idea of physical layer network coding for the two way relay channel was first introduced in \cite{ZLL}, where the multiple access interference occurring at the relay was exploited so that the communication between the end nodes can be done using a two stage protocol. Information theoretic studies for the physical layer network coding scenario were reported in \cite{KMT},\cite{PoY}. The design principles governing the choice of modulation schemes to be used at the nodes for uncoded transmission were studied in \cite{APT1}. An extension for the case when the nodes use convolutional codes was done in \cite{APT2}. A multi-level coding scheme for the two-way relaying scenario was proposed in \cite{HeN}.

It was observed in \cite{APT1} that for uncoded transmission, the network coding map used at the relay needs to be changed adaptively according to the channel fade coefficients, in order to minimize the impact of the multiple access interference. The Latin Square scheme was studied in \cite{NVR}, \cite{VNR} for two way relaying using $M$-PSK signal sets at the end nodes.

\subsection{Signal Model}
\subsubsection*{Multiple Access (MA) Phase}
Let $\mathcal{S}$ denote the square $M$-QAM constellation used at A and B, where $M=2^{2\lambda}$, $\lambda$ being a positive integer. Assume that A (B) wants to transmit a $2\lambda$-bit binary tuple to B (A). Let $\nu :\mathbb{F}_{2^{2\lambda}}   \rightarrow \mathcal{S}$ denote the mapping from bits to complex symbols used at A and B. Let $\nu(s_A)= x_A$, $\nu(s_B)=x_B \in \mathcal{S}$ denote the complex symbols transmitted by A and B respectively, where $s_A,s_B \in \mathbb{F}_{2^{2\lambda}}$. The received signal at $R$ is given by,
\begin{align}
\nonumber
y_R=h_{A} x_A + h_{B} x_B +z_R,
\end{align}
where $h_A$ and $h_B$ are the fading coefficients associated with the A-R and B-R links respectively. The additive noise $z_R$ is assumed to be $\mathcal{CN}(0,\sigma^2)$, which denotes the circularly symmetric complex Gaussian random variable with variance $\sigma ^2$. We assume a block fading scenario, with the ratio $ h_{B}/h_{A}$ denoted as $z=\gamma e^{j \theta}$, where $\gamma \in \mathbb{R}^+$ and $-\pi \leq \theta < \pi$, is referred as the {\it fade state} and for simplicity, also denoted by $(\gamma, \theta)$.
 
 Let $\mathcal{S}_{R}(\gamma,\theta)$ denote the effective constellation at the relay during the MA Phase and $d_{min}(\gamma e^{j\theta})$ denote the minimum distance between the points in $\mathcal{S}_{R}(\gamma,\theta)$, i.e., 
\begin{align} 
\nonumber
 \mathcal{S}_{R}(\gamma,\theta)=\left\lbrace x_i+\gamma e^{j \theta} x_j \vert x_i,x_j \in \mathcal{S}\right \rbrace,
 \end{align}
\begin{align}
%\nonumber
\label{eqn_dmin} 
d_{min}(\gamma e^{j\theta})=\hspace{-0.5 cm}\min_{\substack {{(x_A,x_B),(x'_A,x'_B)}{ \in \mathcal{S}^2 } \\ {(x_A,x_B) \neq (x'_A,x'_B)}}}\hspace{-0.5 cm}\vert \left(x_A-x'_A\right)+\gamma e^{j \theta} \left(x_B-x'_B\right)\vert.
\end{align}

 From \eqref{eqn_dmin}, it is clear that there exists values of $\gamma e^{j \theta}$ for which $d_{min}(\gamma e^{j\theta})=0$. Let $\mathcal{H}=\lbrace \gamma e^{j\theta} \in \mathbb{C} \vert d_{min}(\gamma,\theta)=0 \rbrace$. The elements of $\mathcal{H}$ are said to be {\it the singular fade states} \cite{NVR}. The set $\mathcal{H}$ depends on the signal set used. For example when $\gamma e^{j \theta}=(1+j)/2$, the effective constellation $\mathcal{S}_{R}(\gamma,\theta)$ has only 12 ($<$16) points when 4-QAM signal set is used at nodes A and B. Hence $(1+j)/2  \in \mathcal{H}$ for 4-QAM.

Let $(\hat{x}_A,\hat{x}_B) \in \mathcal{S}^2$ denote the Maximum Likelihood (ML) estimate of $({x}_A,{x}_B)$ at R based on the received complex number $y_{R}$, i.e.,
 \begin{align}
 (\hat{x}_A,\hat{x}_B)=\arg\hspace{-0.5 cm}\min_{({x}'_A,{x}'_B) \in \mathcal{S}^2} \vert y_R-h_{A}{x}'_A-h_{B}{x}'_B\vert.
 \end{align}
 %%%%%%%%%%%%%%%%%%%%%%%%%%%%%%%%%%%%%%%%%%%%%%%%%%%%
%%%%%%%%%%%%%%%%%%%%%%%%%%%%%%%%%%%%%%%%%%%%%
\subsubsection*{Broadcast (BC) Phase}

Depending on the value of $\gamma e^{j \theta}$, R chooses a map $\mathcal{M}^{\gamma,\theta}:\mathcal{S}^2 \rightarrow \mathcal{S}'$, where $\mathcal{S}'$ is the signal set (of size between $M$ and $M^2$) used by R during $BC$ phase. The elements in $\mathcal{S}^2 $ which are mapped on to the same complex number in $\mathcal{S}'$ by the map $\mathcal{M}^{\gamma,\theta}$ are said to form a cluster. Let $\lbrace \mathcal{L}_1, \mathcal{L}_2,...,\mathcal{L}_l\rbrace$ denote the set of all such clusters. The formation of clusters is called clustering, and the set of all clusters is denoted by $\mathcal{C}^{\gamma,\theta}$ to indicate that it is a function of $\gamma e^{j \theta}.$ The received signals at A and B during the BC phase are respectively given by,
\begin{align}
y_A=h'_{A} x_R + z_A,\;y_B=h'_{B} x_R + z_B,
\end{align}
where $x_R=\mathcal{M}^{\gamma,\theta}(\hat{x}_A,\hat{x}_B) \in \mathcal{S'}$ is the complex number transmitted by R. The fading coefficients corresponding to the R-A and R-B links are denoted by $h'_{A}$ and $h'_{B}$ respectively and the additive noises $z_A$ and $z_B$ are $\mathcal{CN}(0,\sigma ^2$).

In order to ensure that A (B) is able to decode B's (A's) message, the clustering $\mathcal{C}^{\gamma,\theta}$ should satisfy the exclusive law \cite{APT1}, i.e.,

{\footnotesize
\begin{align}
\left.
\begin{array}{ll}
\nonumber
\mathcal{M}^{\gamma,\theta}(x_A,x_B) \neq \mathcal{M}^{\gamma,\theta}(x'_A,x_B), \; \mathrm{for} \;x_A \neq x'_A \; \mathrm{,} \; \forall x_B \in  \mathcal{S},\\
\nonumber
\mathcal{M}^{\gamma,\theta}(x_A,x_B) \neq \mathcal{M}^{\gamma,\theta}(x_A,x'_B), \; \mathrm{for} \;x_B \neq x'_B \; \mathrm{,} \;\forall x_A \in \mathcal{S}.
\end {array}
\right\} \\
\label{ex_law}
\end{align}}
The cluster distance between a pair of clusters $\mathcal{L}_i$ and $\mathcal{L}_j$ is the minimum among all the distances calculated between the points $x_A+\gamma e^{j\theta} x_B ,x'_A+\gamma e^{j\theta} x'_B \in \mathcal{S}_R(\gamma,\theta)$ where $(x_A,x_B) \in \mathcal{L}_i$ and $(x'_A,x'_B) \in \mathcal{L}_j$ \cite{NVR}. The \textit{minimum cluster distance} of the clustering $\mathcal{C}$ is the minimum among all the cluster distances, i.e.,

{\footnotesize
\begin{align}
\nonumber
d_{min}^{\mathcal{C}}(\gamma e^{j \theta})=\hspace{-0.8 cm}\min_{\substack {{(x_A,x_B),(x'_A,x'_B)}\\{ \in \mathcal{S}^2,} \\ {\mathcal{M}^{\gamma,\theta}(x_A,x_B) \neq \mathcal{M}^{\gamma,\theta}(x'_A,x'_B)}}}\hspace{-0.8 cm}\vert \left( x_A-x'_A\right)+\gamma e^{j \theta} \left(x_B-x'_B\right)\vert.
\end{align}
}
The minimum cluster distance determines the performance during the MA phase of relaying. The performance during the BC phase is determined by the minimum distance of the signal set $\mathcal{S}'$. For values of $\gamma e^{j \theta}$ in the neighbourhood of the singular fade states, the value of $d_{min}(\gamma e^{j\theta})$ is greatly reduced, a phenomenon referred as {\it distance shortening}. To avoid distance shortening, for each singular fade state, a clustering needs to be chosen such that the minimum cluster distance at the singular fade state is non-zero and is also maximized.  

A clustering $\mathcal{C}$ is said to remove a singular fade state $ h \in \mathcal{H}$, if $d_{min}^{\mathcal{C}}(h)>0$. 
For a singular fade state $h \in \mathcal{H}$, let $\mathcal{C}^{\lbrace h\rbrace}$ denote a clustering which removes the singular fade state $h$ (if there are multiple clusterings which remove the same singular fade state $h$, consider a clustering which maximizes the minimum cluster distance). Let $\mathcal{C}_{\mathcal{H}}=\left\lbrace \mathcal{C}^{\lbrace h\rbrace} : h \in \mathcal{H} \right\rbrace$ denote the set of all such clusterings. Let $d_{min}({\mathcal{C}^{\lbrace h\rbrace}},\gamma',\theta')$ be defined as,

{\footnotesize
\begin{align}
\nonumber
d_{min}({\mathcal{C}^{\lbrace h\rbrace}},\gamma',\theta')=\hspace{-1 cm}\min_{\substack {{(x_A,x_B),(x'_A,x'_B)} \\ { \in \mathcal{S}^2,} \\ {\mathcal{M}^{\lbrace h\rbrace}(x_A,x_B) \neq \mathcal{M}^{\lbrace h \rbrace}(x'_A,x'_B)}}}\hspace{-0.9 cm}\vert \left( x_A-x'_A\right)+\gamma' e^{j \theta'} \left(x_B-x'_B\right)\vert.
\end{align}
}

The quantity $d_{min}({\mathcal{C}^{\lbrace h\rbrace}},\gamma,'\theta')$ is referred to as  the minimum cluster distance of the clustering $\mathcal{C}^{\lbrace h\rbrace}$ evaluated at $\gamma' e^{j\theta'}.$

In practice, the channel fade state need not be a singular fade state. In such a scenario, among all the clusterings which remove the singular fade states, the one which maximizes the minimum cluster distance is chosen. In other words, for $\gamma' e^{j \theta'} \notin \mathcal{H}$, the clustering $\mathcal{C}^{\gamma',\theta'}$ is chosen to be $\mathcal{C}^{\lbrace h\rbrace}$, which satisfies $d_{min}({\mathcal{C}^{\lbrace h\rbrace}},\gamma',\theta') \geq d_{min}({\mathcal{C}^{\lbrace h' \rbrace}},\gamma',\theta'), \forall h \neq h' \in \mathcal{H}$. Since the clusterings which remove the singular fade states are known to all the three nodes and are finite in number, the clustering used for a particular realization of the fade state can be indicated by R to A and B using overhead bits.

%For $\gamma e^{j \theta} \notin \mathcal{H}$, the clustering $\mathcal{C}$ is chosen to be $\mathcal{C}_{\lbrace h\rbrace}$, which satisfies $d_{min}^{\mathcal{C}_{\lbrace h\rbrace}}(\gamma e^{j \theta}) \geq d_{min}^{\mathcal{C}_{\lbrace h' \rbrace}}(\gamma e^{j \theta}), \forall h \neq h' \in \mathcal{H}$.
%%%%%%%%%%%%%%%%
\begin{example}
In the case of BPSK, if channel condition is $\gamma=1$ and $\theta=0$ the distance between the pairs $(0,1)(1,0)$ is zero as in Fig.\ref{fig:BPSK}(a). The following clustering removes this singular fade state.
$$\{\{(0,1)(1,0)\},\{(1,1)(0,0)\}\}$$
The minimum cluster distance is non zero for this clustering.
%%%%%%%%%%%%%%%% 
 \begin{figure}[h]
\centering
\vspace{-1.05 cm}
\includegraphics[totalheight=2.5in,width=3.5in]{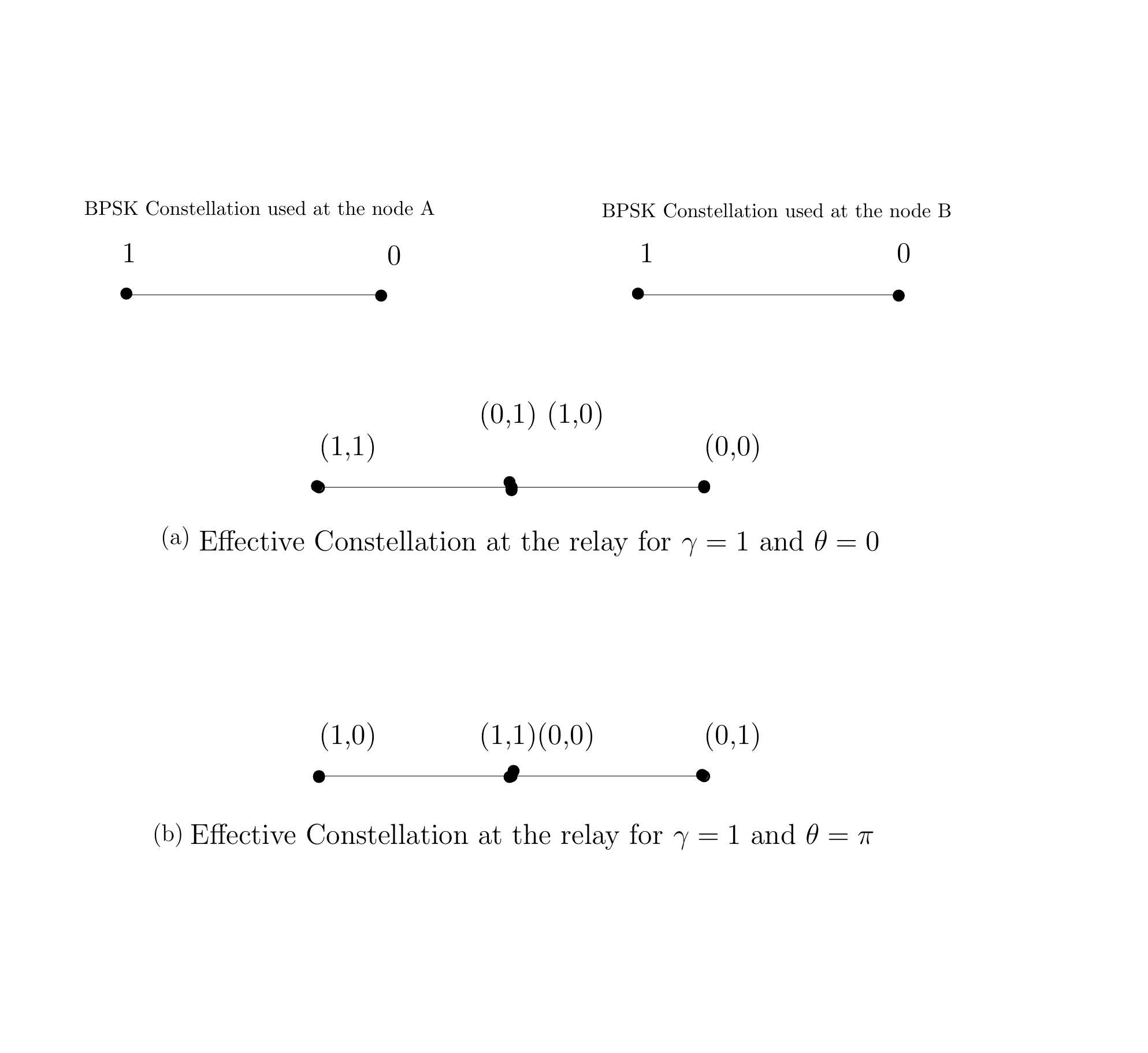}
\vspace{-1.5 cm}
\caption{Effective Constellation at the relay for singular fade states, when the end nodes use BPSK constellation.}     
\label{fig:BPSK}        
\end{figure}
\end{example}
%%%%%%%%%%%%%%%%%%%%%%%%%%%%%%%%%%%%%%%%%%%%%%%%%%
 %\vspace{-.25 in}
%%%%%%%%%%%%%%%%%%%%%%%%%%%%%%%%%%%%%%%%%%%%%%%%
The contributions and organization of the paper are as follows:
\begin{itemize}
\item A procedure to obtain the number of singular fade states for PAM and QAM signal sets is presented.  
\item It is shown that for the same number of signal points $M$, the number of singular fade states for square $M$-QAM is lesser than the number of singular fade states for $M$-PSK. The advantages of this are two fold - QAM offers better distance performance in the MA Phase and QAM requires lesser number of Latin squares (i.e., a reduction in number of overhead bits).
\item It is known from \cite{VR} that the removal of the singular fade state $z$=1 assumes greater significance in a Rician fading scenario. The bit-wise XOR map removes this singular fade state when $M$-PSK signal set is used at nodes A and B for any $M$. It is shown that XOR mapping cannot remove this singular fade state for any  $M$-QAM and a different mapping is obtained to remove this singular fade state for $M$-QAM.
\item Inspired from \cite{VNR}, the problem of partitioning the entire complex plane into clustering independent region as well as clustering dependent region is addressed. The approach followed for $M$-QAM signal set needs to be different from that used for $M$-PSK in \cite{VNR}.
\item The region associated with each singular fade state in the complex plane is obtained analytically for $M$-QAM signal set used at nodes A and B. This helps in associating a Latin Square corresponding to that singular fade state to the said region like in $M$-PSK. By simulation it is shown that the choice of 16-QAM leads to better performance for both the Rayleigh and the Rician fading scenario, compared to 16-PSK.
\end{itemize}
The remaining content is organized as follows:  

In Section \ref{sec2}, we discuss the relationship between singular fade states and difference constellation of the signal sets used by the end nodes. We present expressions to get the number of singular fade states for PAM and square QAM signal sets in Subsections \ref{subsec_1_2} and \ref{subsec_2_2} respectively. In Subsection \ref{subsec_3_2}, it is proved that the number of singular fade states for $M$-QAM is always lesser than that of $M$-PSK signal sets.  In Section \ref{sec3}, the clustering for singular fade states is obtained through completing Latin Squares and a Latin Square for removing the singular fade state $z=1$ is analytically obtained for PAM and QAM signal sets. In Section \ref{sec4}, channel quantization for $M$-QAM signal set is discussed. In particular, the channel quantization for the entire complex plane when nodes A and B use 16-QAM signal set is obtained. In Section \ref{sec5}, simulation results are provided to show the advantage of Latin Square scheme for QAM over XOR network coding scheme as well as Latin Square scheme for PSK signal sets under Rayleigh and Rician fading channel assumptions.

%\vspace{-.65 cm}
 %%%%%%%%%%%%%%%%%%%%%%%%%%%%%%%%%%%%%%%%%%%%%%%%%%%%%%%%
 \section{Singular Fade states and Difference Constellations}
%\vspace{-.15 cm} 
 \label{sec2}
%In this section we show the relationships between singular fade states and difference constellation of the signal set used by the end nodes. %Throughout in our discussion we exclude the singular fade states $z=0$ and $z=\infty$ since they are irremovable singular fade states.
The location of singular fade states in the complex plane for any constellation used at end nodes can be characterised in the following way.
%%%%%
 If node A uses a constellation $\mathcal{S}_1$ of size $M_1$ and node B, a constellation $\mathcal{S}_{2}$ of size $M_2$, the singular fade states $z=\gamma e^{j \theta}$ are of the form
 \begin{equation}
 \label{sing_expression}
 z=\gamma e^{j\theta}=\dfrac{x_A-x_A^{\prime}}{x_B^{\prime}-x_B} 
\end{equation}
 and is obtained by equating $x_A+\gamma e^{j\theta} x_B$ and $x'_A+\gamma e^{j\theta} x'_B$ for $x_A,x_A^{\prime} \in \mathcal{S}_1$ and $x_B,x_B^{\prime} \in \mathcal{S}_2$. Henceforth, throughout the paper,  we assume  both the end nodes use the same constellation, $\mathcal{S}$. 

%%%%%%%%%%%%%%%%%%%%%%%%%%%%%%%%%%%%%%%%%%%%%%%%%%%%%%%%%%%%%
\subsection{Singular Fade States of PAM signal sets} 
\label{subsec_1_2}
 %%%%%%%%%%%%%%%%%%%%%%%%%%%%%%%%%%%%%%%%%%%%%%%%%%%%%%%%%%%%%%
Consider the symmetric $\sqrt M$-PAM signal set given by 
$$ \mathcal{S}= \left\lbrace -(\sqrt M -1)+ 2n \right\rbrace,  ~~~  n \in (0, \cdots ,\sqrt M -1). $$ 
The difference constellation of $\mathcal{S}$ is 
$$ \Delta\mathcal{S}= \left\lbrace x-x' : x,x' \in {\mathcal{S}} \right\rbrace $$ and can be written in the form
$$ \Delta\mathcal{S}= \left\lbrace-2(\sqrt M -1)+ 2n \right\rbrace, ~~~ n \in (0, \cdots ,2(\sqrt M -1)).$$ For example, the 4-PAM signal set and it's difference constellation are given in Fig.\ref{fig:pam} and Fig.\ref{fig:pamdiff} respectively. For each of the difference constellation points, the pair in the signal set which corresponds to this point is also shown.
%%%%%%%%%%%%%%%%%%%%%%%%%%%%%%%%%%%%%%%%%%%
 \begin{figure}[t]
\centering
%\vspace{-.25 in}
\subfigure[$\sqrt M$ PAM constellation]{
\includegraphics[totalheight=.4in,width=2in]{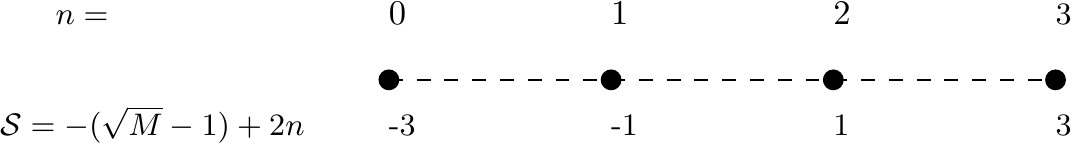}
\label{fig:pam}
}

\subfigure[Difference Constellation]{
\includegraphics[totalheight=.8in,width=3.25in]{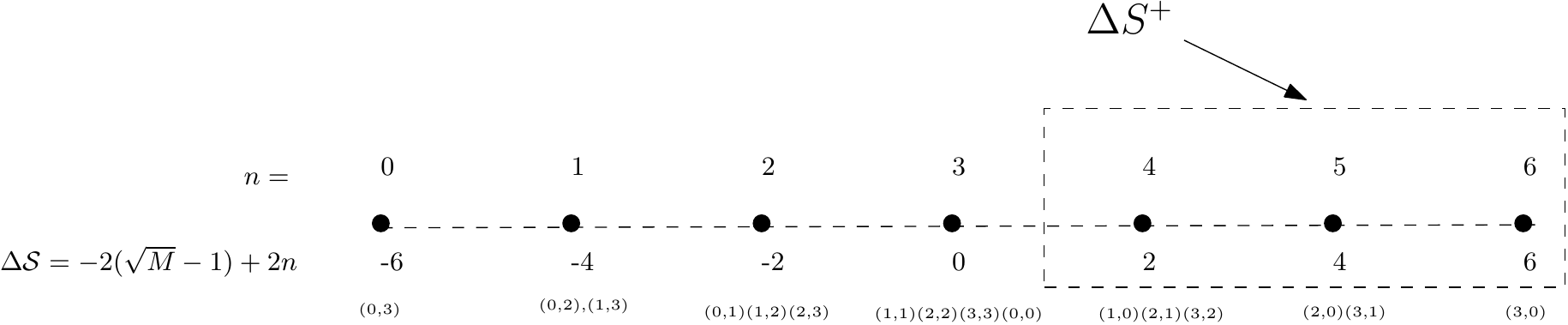}
\label{fig:pamdiff}
}

\caption{$\sqrt M$ PAM constellation and difference constellation for $\sqrt M=4$}
\label{fig:4pam}
\vspace{-.5 cm}
\end{figure}
%%%%%%%%%%%%%%%%%%%%%%%%%%%%%%%%%%%%%%%%%%%%%%%%
We will often consider only the first quadrant of $\Delta{S}$, denoted as $\Delta{S}^{+},$ which for a general complex signal set is given by  $$\Delta{S}^{+} = \{\alpha : \mbox{real}(\alpha)>0, \mbox{imaginary}(\alpha) \geq 0\}.$$ %From Definition \ref{sfs_alter_def}, it can be seen that the number of distinct singular fade states is equal to the number of different mappings $\mathcal{Z}$ possible. 
The following lemma gives the number of singular fade states for PAM signal sets.
 
 %%%%%%%%%%%%%%%%%%%%%%%%%%%%%%%%%%%%%%%%%%%%%%%
 \begin{lemma}
 \label{no_sing_pam}
 The number of singular fade states, for a regular $\sqrt M$-PAM signal set, denoted by $N_{(\sqrt M-PAM)}$ is given by
 \begin{equation}
 \label{sum_euler}
 N_{(\sqrt M-PAM)}= 2 + 4\sum_{n=1}^{\sqrt M -1}n \prod_{p\vert n} \left(1-\frac{1}{p}\right),
 \end{equation}
 where $p|n$ stands for  prime number $p$ dividing $n.$
 \end{lemma} 

%%%%%%%%%%%%%%%%%%%%%%%%%%%%%%%%
\begin{proof}
See Appendix \ref{app2}.
\end{proof}
%%%%%%%%%%%%%%%%%%%%%%%%%%%%%%%%%%
\begin{example}
Consider the case of 4-PAM ($M=16$) signal set as given in Fig.\ref{fig:4pam}. There are $2(\sqrt M-1)=6$ non-zero signal points in the difference constellation. Scaled $\Delta{S}^+$ has $(\sqrt M-1)=3$ signal points-$\{1,2,3\}$. And there are 14 singular fade states-
\begin{align*}
\left\{1,\frac{1}{2},\frac{1}{3},\frac{2}{3},2,3,\frac{3}{2},-1,\frac{-1}{2}, \frac{-1}{3},\frac{-2}{3},-2,-3,\frac{-3}{2}\right\}.
\end{align*} 
Calculating \eqref{sum_euler} also gives $N_{(4-PAM)} =14.$ Calculating similarly, we find that for 8-PAM signal set, there exists 70 singular fade states.
\end{example}
%%%%%%%%%%%%%%%% 
%  \begin{figure}[htbp]
% \centering
% \vspace{-.4 cm}
% \includegraphics[totalheight=2in,width=3.5in]{pam_sing.eps}
% \vspace{-2.6 cm}
% \caption{Singular fade states of 4-PAM signal set}     
% \label{fig:pam_sing}        
% \end{figure}

%%%%%%%%%%%%%%%%%%%%%%%%%%%%%%%%%%%%%%%%%%%%%%%%%%

%%%%%%%%%%%%%%%%%%%%%%%%%%%%%%%%%%%%%%%%%%%%%%%%%%%
%%%%%%%%%%%%%%%%%%%%%%%%%%%%%%%%%%%%%%%%%%%
% \begin{figure}[htbp]
%\centering
%\vspace{-.25 in}
%\subfigure[$16-$QAM constellation]{
%\includegraphics[totalheight=2in,width=2.5in]{16qam.pdf}
%\label{fig:16qam}
%}
%\subfigure[The Difference Constellation]{
%\includegraphics[totalheight=2.1in,width=2.75in]{16qam_diff_cons.pdf}
%\label{fig:16qam_diff_cons}
%}
%\caption{$16-$QAM constellation and its difference constellation}
%\label{fig:qam}
%\vspace{-.5 in}
%\end{figure}
 %%%%%%%%%%%%%%%%%%%%%%%%%%%%
%\begin{table}
%\centering
% \caption{Singular fade states for 8-PAM}
% \label{prifactor_pam}
%\begin{tabular}{|c|c|c|c|}
%\hline $n$, Elements  & Relative primes & Singular fade & $\psi(n)$\\ 
% in $\Delta{S}^{+}$ & less than $n$  & states, $z >1$& \\
%\hline 1 &   &   & 0\\ 
%\hline 2 & 1 & 2 & 1\\
%\hline 3 & 1,2 & 3,$\frac{3}{2}$ & 2\\
%\hline 4 & 1,3 & 4,$\frac{4}{3}$ & 2\\
%\hline 5 & 1,2,3,4 & 5,$\frac{5}{2},\frac{5}{3},\frac{5}{4}$ & 4\\
%\hline 6 & 1,5 & 6,$\frac{6}{5}$ & 2\\
%\hline 7 & 1,2,3,4,5,6 & 7,$\frac{7}{2},\frac{7}{3},\frac{7}{4},\frac{7}{5},\frac{7}{6}$ & 6\\
%\hline 
%\end{tabular}
%\vspace{-.1 cm}
%\end{table}
%%%%%%%%%%%%%%%%
%%%%%%%%%%%%%%%%%%%%%%%%%%%%%%%%%%%%%%%%%%%%%%%%%%
%%%%%%%%%%%%%%%%%%%%%%%%%%%%%%%%%%%%%%%%%%%%%%%%%%
\subsection{Singular Fade States for QAM signal sets}
\label{subsec_2_2}
We consider square $M$-QAM signal set $\mathcal{S}=\{A_{mI}+jA_{mQ}\}$, where $A_{mI}$ and $A_{mQ}$ take values from the $\sqrt M$-PAM signal set $ -(\sqrt M -1)+ 2n,  ~~~  n \in (0, \cdots ,\sqrt M -1).$ We use the bijective mapping $\mu: \mathcal{S} \rightarrow \mathbb{Z}_M = \{0, 1, \cdots, M-1\}$ given by 
{\small
\begin {equation}
\label{mumap}
A_{mI}+jA_{mQ} \rightarrow \frac{1}{2}[(\sqrt M -1 +A_{mI})\sqrt M + (\sqrt M -1 +A_{mQ})]
\end{equation}
}for concreteness, even though our analysis and results hold for any map.  The difference constellation $\Delta\mathcal{S}$ of square $M$-QAM signal set forms a part of scaled integer lattice with $(2\sqrt M -1)^2$ points. The 4-QAM signal set with the above mapping and its difference constellation are shown in Fig.\ref{fig:4qam} and in Fig.\ref{fig:4qamdiffcons}.

In a practical scenario, there can be an average energy constraint $E$ to be satisfied at nodes A and B. In such a case, we use a scaled version of the $M$-QAM signal set given by $\frac{1}{\sqrt{\rho}}\{A_{mI}+jA_{mQ}\}$, where $\rho$ is chosen so as to meet the constraint $E$. As a special case, for $E=1$ (unit normalisation),  $\rho=\frac{2}{3(M-1)}$. It may be noted here that the values of the singular fade states (for a particular choice of $\mathcal{S}$) are unaffected by the energy constraint $E$. However, the minimum distances of the constellation $\mathcal{S}$ and the effective constellation ${\mathcal{S}}_R(\gamma,\theta)$ are dependent on the choice of $E$. In particular, for unit normalisation,
\begin{align}
\label{eqndmin}
d_{min}(\mbox{$M$-QAM}) &= \frac{2}{\sqrt{\rho}}
= \sqrt{\frac{6}{M-1}}.
\end{align}
Compare this to the case with $M$-PSK whose \mbox{$d_{min} = 2\sin(\pi/M)$}. For \mbox{$M=16$}, \mbox{$d_{min}(\mbox{16-QAM}) = \sqrt{\frac{2}{5}}$} \mbox{$> d_{min}(\mbox{16-PSK}) = 2\sin(\pi/16)$}. This has a detrimental effect in the performance during the MAC phase.
%%%%%%%%%%%%%%%%%%%%%%%
\begin{figure}[htbp]
\centering
%\vspace{-.5cm}
\subfigure[$4-$QAM constellation]{
\includegraphics[totalheight=1.1in,width=1.1in]{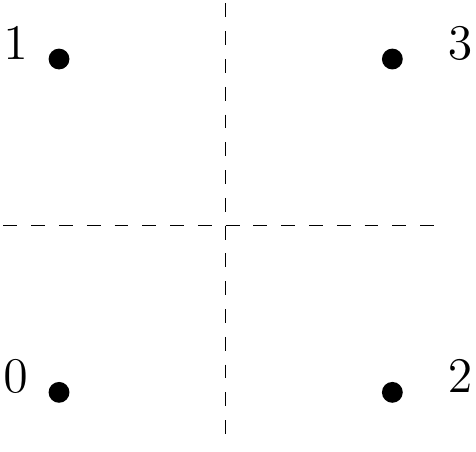}
\label{fig:4qam}
}
\subfigure[The Difference Constellation]{
\includegraphics[totalheight=2.2in,width=2.2in]{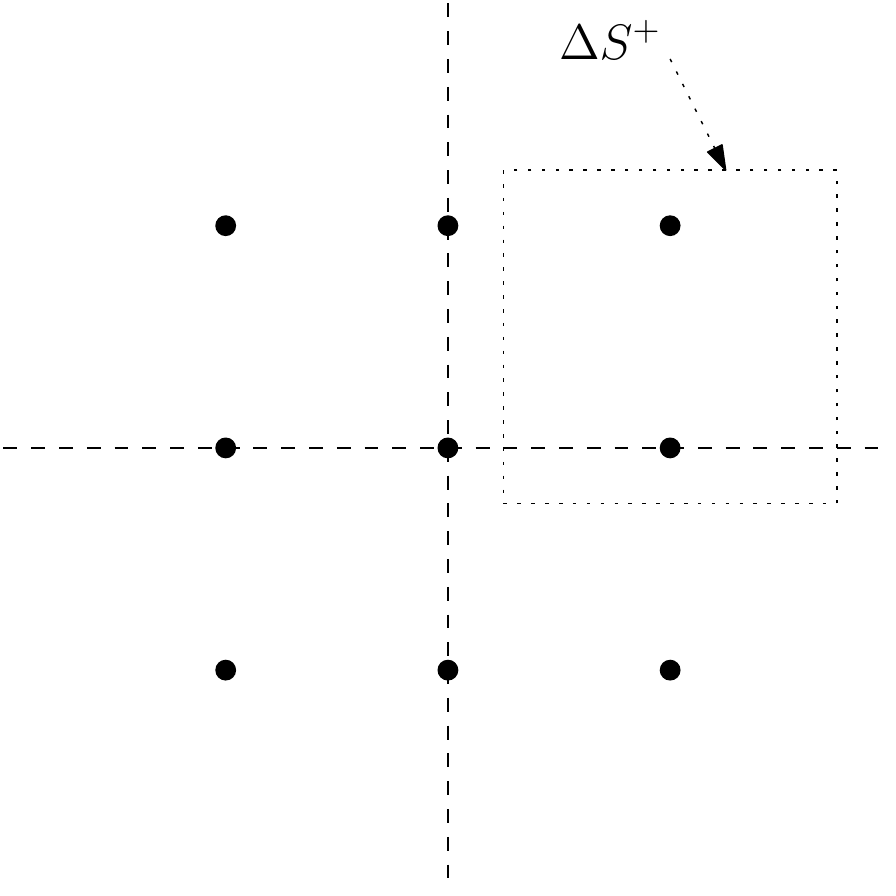}
\label{fig:4qamdiffcons}
}
\subfigure[Singular fade states]{
\includegraphics[totalheight=1.8in,width=2in]{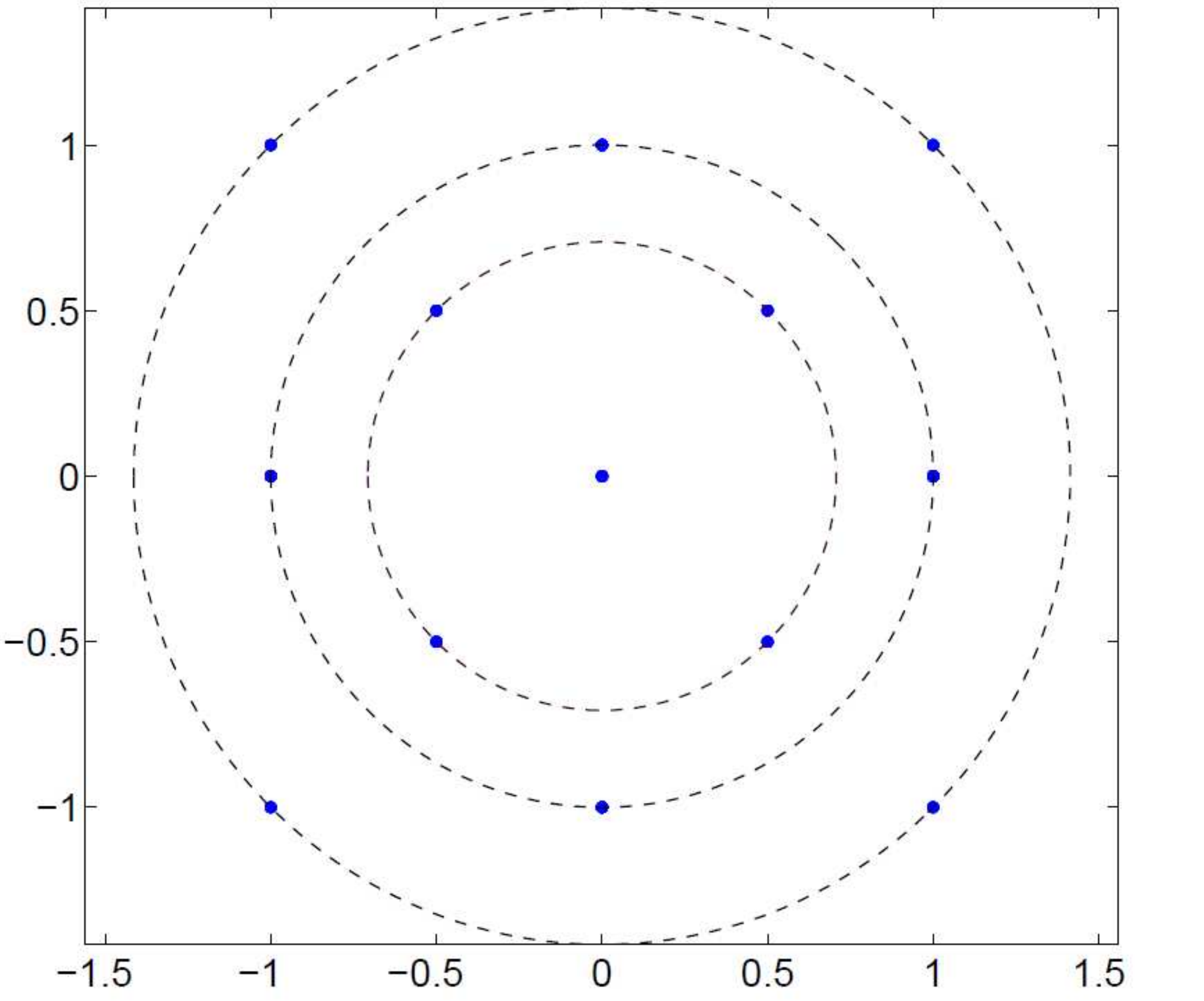}
\label{fig:4qam_sfs}
}
\caption{$4-$QAM constellation, its  difference constellation and singular fade states}
\label{fig:4QAM}
%\vspace{-.5 in}
\end{figure}

%%%%%%%%%%%%%%%%%%%%%%%%%%%%%%%%%%%%%%%%%%%%%%%%
%%%%%%%%%%%%%%%%%%%%%%%%%%%%%%%%%%%%%%%%%%%%%%%%

The signal points in the difference constellation  of $M$-QAM are Gaussian integers scaled by $\sqrt{\rho}$. To get the number of singular fade states for square QAM signal sets, the notion of primes and relatively primes in the set of Gaussian integers is useful.
\begin{definition}
\cite{LB} The Gaussian integers are the elements of the set $\mathbb{Z}[j]=\{a+bj : a,b \in \mathbb{Z}\}$, where $\mathbb{Z}$ denotes the set of integers. 
A Gaussian integer $\alpha$ is called a Gaussian prime if the Gaussian integers that divide $\alpha$ are: $1,-1,j,-j,\alpha,-\alpha, \alpha j$ and $-\alpha j.$ The Gaussian integers which are invertible in $\mathbb{Z}[j]$ are called units in $\mathbb{Z}[j]$ and they are $\pm 1$ and $\pm j.$  Let $\alpha, \beta \in \mathbb{Z}[j]$. If the only common divisors of $\alpha$ and $\beta$ are units, we say $\alpha$ and $\beta$ are relatively primes.
\end{definition}

\begin{lemma}
\label{no_sfs_qam}
The number of singular fade states for the square $M$-QAM signal set, denoted by
$N_{M-QAM}$ is given by $$N_{M-QAM} = 4+ 8 \phi(\Delta{S}^{+})$$ where $\phi(\Delta{S}^{+})$ is the number of relative prime pairs in $\Delta{S}^{+}.$
\end{lemma}

\begin{proof}
All the possible ratios of elements from $\Delta{S}$ give singular fade states. We consider only ratios in $\Delta{S}^+$ and multiply the number of such possible ratios by a factor of 4 to account the ratios with points in all the other quadrants. To avoid multiplicity while counting, we take only relative prime pairs in $\Delta{S}^+$ and every such pair $(a,b)$ gives two singular fade states $a/b$ and $b/a$. Because of this, the multiplication factor becomes 8. Finally, the constant term 4 is added to count the units.
\end{proof}
%%%%%%%%%%%%%%%%%%%%%%%%%%%%%%%%%%%%%%%%%%%%%%%%%
\begin{example}
For 4-QAM signal set shown in Fig.\ref{fig:4qam}, the number of singular fade states, $N_{4-QAM}$ is 12. Scaled $\Delta{S}^{+}$ has only two elements $\{1,1+j\}$ in this case, as shown in Fig.\ref{fig:4qamdiffcons}. They form a relatively prime pair. The singular fade states $\pm1,\pm j,\pm1\pm j,\frac{1}{\pm1\pm j}$ are shown in Fig. \ref{fig:4qam_sfs}.
\end{example}
%%%%%%%%%%%%%%%%%%%%%%%%%%%%%%%%%%%%%%%%%%%%%%%%%
%\begin{table}
%\centering
% \caption{Prime factors of Gaussian integers in $\Delta{S}^+$}
% \label{prifactor}
%\begin{tabular}{|c|c|c|}
%\hline Elements in $\Delta{S}^{+}$ & Prime factors  & No. of relatively prime pairs\\ 
%\hline 1 &  1 & 11\\ 
%\hline 1+j& 1+j & 6\\
%\hline 2& 1+j & 6\\
%\hline 1+2j& 1+2j & 10\\
%\hline 2+j& 2+j & 10\\
%\hline 2+2j& 1+j & 6\\
%\hline 3& 3 & 10\\
%\hline 3+j& (1+j),(1+2j)& 5\\
%\hline 1+3j& (1+j),(2+j)& 5\\
%\hline 3+2j& 3+2j & 11\\
%\hline 2+3j& 2+3j & 11\\
%\hline 3+3j& (1+j),3 & 5\\
%%\hline Total No. of relatively prime pairs for $\Delta{S}^+$ for 16-QAM & & 96 \\
%\hline
%\end{tabular}
%\vspace{-.1 cm}
%\end{table}
%%%%%%%%%%%%%%%% 

\begin{example}
Consider the case of 16-QAM signal set. It can be verified that there are 48 distinct pairs of relative primes, and from Lemma \ref{no_sfs_qam}, $N_{16-QAM}$ turns out to be 388. 
\end{example}
%%%%%%%%%%%%%%%%%%%%%%%%%%%%%%%%%%%%%%%%%%%%%%%%%%%%%%%%%
\begin{table}[htbp]
\centering
 \caption{Comparison between $M$-PSK and $M$-QAM on number of singular fade states}
 \label{comp_psk_qam}
 % \vspace{-.15 in}
\begin{tabular}{|c|c|c|}
\hline $M$  & $M$-PSK &  $M$-QAM\\
\hline 4 &  12 & 12\\ 
\hline 16 & 912 & 388\\
\hline 64 & 63,552 & 8388\\
\hline
\end{tabular}
%\vspace{-.8 cm}
\end{table}
%%%%%%%%%%%%%%%%%%%%%%%%%%%%%%%
%%%%%%%%%%%%%%%%%%%%%%%%%%%%%%%%%%%%%%%%%%%%%%%%%%%
\subsection{Singular fade states of $M$-PSK and $M$-QAM signal sets}
\label{subsec_3_2}
In this section we show that the number of singular fade states for $M$-QAM signal sets is lesser than that of $M$-PSK signal sets. The advantages of this are two fold- QAM offers better distance performance during the MA phase and it requires lesser number of overhead bits during the BC phase, since the required number of relay clusterings is lesser in the case of QAM compared to PSK.
%%%%%%%%%%%%%%%%%%%%%%%%%%%%%%%%%%%%%%%%%%%%%%%%%%%%
\begin{lemma}
\label{upper_qam}
The number of singular fade states for $M$-QAM signal set is upper bounded by $4(n^2-n+1)$, where $n=\frac{1}{4}{[(2 \sqrt M -1)^2-1]}$, which is same as $4M^2-(2M-1)\sqrt{M} +1).$ 
\end{lemma}
\begin{proof}
There are $[(2 \sqrt M -1)^2-1]$ non zero signal points in $\Delta{S}$ which are distributed equally in each quadrant, i.e., the number of signal points in  $\Delta{S}^+$ is $\frac{1}{4}{[(2 \sqrt M -1)^2-1]}$ which we denote by $n$.  %From Definition \ref{sfs_alter_def} it can be seen that the number of singular fade states is related to the number of different mappings possible with the $n$ signal points. 
The maximum number of relatively prime pairs in a set of $n$ Gaussian integers is $\frac{n(n-1)}{2}$. Since an upper bound is of interest we substitute this in Lemma \ref{no_sfs_qam} instead of $\phi(\Delta{S}^+).$ This completes the proof.
\end{proof}
%%%%%%%%%%%%%%%%%%%%%%%%%%%%%%%%%%%%%%%%%%%%%%%%%%%%%%%

In \cite{VNR}, it is shown that the number of singular fade states for $M$-PSK signal set is $M(\frac{M^2}{4}-\frac{M}{2}+1)$, of $\mathcal{O}(M^3)$. From Lemma \ref{upper_qam}, an upper bound on the number of singular fade states for $M$-QAM is of $\mathcal{O}(M^2).$ Hence, the number of singular fade states for $M$-QAM signal set is lesser than that of $M$-PSK signal sets. 

%%%%%%%%%%%%%%%%%%%%%%%%%%%%%%%%%%%%%%%%%%%%%%%%%%%%%%%
%\begin{example}
%The singular fade states of 16-PSK signal set is given in Fig.\ref{fig:16psk_sing}. There are 912 singular fade states in total.
%%\begin{figure}[h]
%%\centering
%%%\vspace{-.8 cm}
%%\includegraphics[totalheight=2in,width=3.5in]{16psk_sing.eps}
%%%\vspace{-2 cm}
%%\caption{Singular fade states for 16-PSK signal sets.}     
%%\label{fig:16psk_sing}        
%%\end{figure}
%\end{example}
%%%%%%%%%%%%%%%%%%%%%%%%%%%%%%%%%%%%%%%%%%%%%%%%%%%%%%%%%

The advantage of using square QAM constellation is more significant in higher order constellations as shown in Table \ref{comp_psk_qam}. For instance, 64-QAM has 8,388 singular fade states where as 64-PSK has 63,552 singular fade states and the relay has to adaptively use 63,552 clusterings. So, with the use of square QAM constellations the complexity is enormously reduced. 
% \vspace{-.25 in}
%%%%%%%%%%%%%%%%%%%%%%%%%%%%%%%%%%%%%%%%%%%%%%%%%
%\begin{figure*}[]
%\centering
%\hspace{2cm}
%\subfigure[16-QAM ]{
%\includegraphics[totalheight=2.3in,width=2.3in]{16QAM_sing.eps}
%\label{fig:16QAM_sing}
%}
%\qquad
%\hspace{-1.5cm}
%\subfigure[ 16-PSK ]{
%\includegraphics[totalheight=2.25in,width=4in]{16psk_sing.eps}
%\label{fig:16psk_sing}
%}
%\caption{Singular Fade States for $16-$QAM and $16-$PSK modulation schemes}
%%\label{fig:16psk_sing}
%\end{figure*}

%%%%%%%%%%%%%%%%%%%%%%%%%%%%%%%%%%%%%%%%%%%%%%%%%%
%%%%%%%%%%%%%%%%%%%%%%%%%%%%%%%%%%%
\section{Exclusive Law and Latin Squares}
% \vspace{-.15 in}
\label{sec3}
\begin{definition} \cite{Rod} A Latin Square $L$ of order $M$ with the symbols from the set $\mathbb{Z}_t=\{0,1, \cdots ,t-1\}$ is an \textit{M} $\times$ \textit{M}  array, in which each cell contains one symbol and each symbol occurs at most once in each row and column. 
\end{definition}

 In \cite{NVR}, it is shown that when the end nodes use signal sets of the same size, all the relay clusterings which satisfy exclusive-law are equivalently representable by Latin Squares, with the rows (columns) indexed by the constellation point indices used by node A (B) and the clusterings are obtained by placing into the same cluster all the row-column pairs which are mapped to the same symbol in the Latin Square.
 %%%%%%%%%%%%%%%%%%%%%%%%%%%%%%%%%%%%%%%%%%%%%%%%%%%%%% 
 \subsection{Removing Singular fade states and Constrained Latin Squares}
The minimum size of the constellations needed in the  BC phase is $M$, but it is observed that in some cases relay may not be able to remove the singular fade states with $t=M$ and $t > M$ results in severe performance degradation in the MA phase \cite{APT1}. Let $(k,l)$ and $(k^{\prime},l^{\prime})$ be the pairs which give the same point in the effective constellation $\mathcal{S}_R$ at the relay for a singular fade state, where $k,k^{\prime},l,l^{\prime} \in \{0,1,....,M-1\}$ and $k,k^{\prime}$ are the constellation points used by node A and $l,l^{\prime}$ are the corresponding constellation points used by node B. If these are not clustered together, the minimum cluster distance will be zero. To avoid this, such pairs should be in the same cluster. This requirement is termed as {\it singularity-removal constraint} \cite{NVR}. So, we need to obtain Latin Squares which can remove singular fade states and with minimum value for $t$. Towards this end, for a given singular fade state $z \in {\cal{H}}$, initially we fill the $\textit{M}\times\textit{M}$ array in such a way that the slots corresponding to a singularity-removal constraint are filled using the same element. Similarly, we fill in elements for the other singularity removal constraints for the given singular fade state. This removes that particular singular fade state. Such a partially filled Latin Square is called a {\it Constrained Partially-filled Latin Square} (CPLS). After this,  to make this a Latin Square, we try to fill the other slots of the CPLS with minimum number of symbols.
%%%%%%%%%%%%%%%%%%%%%%%%%%%%%%%%%%%%%%%%%%%%%%%%%%%%%%%%
%%%%%%%%%%%%%%%%%%%%%%%%%%%%%%%%%%%%%%%%%%%%

From \cite{NVR} it is known that if the Latin square $L$ removes the singular fade state $z$ then the  Latin Square $L^T$ removes the singular fade state $z^{-1}$, where $L^T$ is the transpose of the Latin Square $L$, i.e. $L^T(i,j)=L(j,i)$ for all $i,j \in \{0,1,2,..,M-1\}$. This observation, in fact, holds for any choice of constellation $\mathcal{S}$ used at nodes A and B.

%%%%%%%%%%%%%%%%%%%%%%%%%%%%%%%%%%%%%%%%%%%%%%%

From this, it is clear that we need to get Latin Squares only for those singular fade states with $|z| \leq 1$ or $|z| \geq 1$. 

The square QAM signal set has a symmetry which is $\pi/2$ degrees of rotation. This results in a reduction of the number of required Latin Squares by a factor 4 as shown in the following lemma. 

\begin{lemma}
\label{iso_pi_bi_two}
If $L$ is a Latin Square that removes a singular fade state $z$, then there exists a column permutation of $L$ such that the permuted Latin Square $L^\prime$ removes the singular fade state $z e^{j \pi/2}.$
\end{lemma}

\begin{proof}
For the singular fade state $z$ as given in \eqref{sing_expression} with constraint $\{(x_A,x_B),(x_A^{\prime},x_B^{\prime})\}$, the singular fade state $z e^{j \pi/2}$ is given by
\begin{align*}
z e^{j \pi/2}= \dfrac{[x_A-x_A^{\prime}]}{[x_B^{\prime}-x_B]}e^{j \pi/2}\\
\Longrightarrow z e^{j \pi/2}= \dfrac{[x_A -x_A^{\prime} ]}{[x_B^{\prime}e^{-j \pi/2}-x_B e^{-j \pi/2}]}.\\
\end{align*}
\noindent
Since in the square QAM constellation there exist signal points with $x_B^\prime e^{-j \pi/2}$ and $x_B e^{-j \pi/2}$, though all the constraints are changed, the new constraints are obtainable by a permutation of signal points in the constellation used by node B. The columns of the Latin Squares are indexed by the signal points used by B and the effected permutation in the constellation is representable by column permutation in the Latin Square.
\end{proof}
%%%%%%%%%%%%%%%%%%%%%%%%%%%%%%%%%%%%%%%%%%%%%%%%%%%%%%%%%%%%%%%%%%%%%%

\begin{lemma}
\label{refl_sym}
If a Latin Square, $L$ removes a singular fade state ($\gamma$, $\theta$) then the Latin Square to remove  ($\gamma$, $(90-\theta)$) is obtainable from $L$ by appropriate row and column permutations. 
\end{lemma}
%%%%%%%%%%%%%%%%%%%%%%%%%%%%%%%%%%
\begin{proof}
See Appendix \ref{app7}.
\end{proof}

%%%%%%%%%?%%%%%%%%%%%%%%%%%%%%%%%%%%%%%%%%%%%%%%%%%%%%%%%%%%%%%%%%%%
Note that the fade state $z=1$ or $(\gamma=1, \theta=0)$ is a singular fade state for any signal set. 
\begin{definition} A Latin Square which removes the singular fade state $z=1$ for a signal set is said to be a standard Latin Square for that signal set. 
\end{definition} 

It is known that the removal of the singular fade state $z$=1 has very high significance in a Rician fading scenario when the Rician factor $K\neq0$. The readers may refer to \cite{VR} for further details regarding the influence of the values of singular fade states on the overall performance of the bi-directional relay network.

When nodes A and B use a $2^\lambda$-PSK signal set, it has been shown in \cite{NVR} that the Latin Square obtained by Exclusive-OR (XOR) is a standard Latin Square for any integer $\lambda$. It turns out that for $M$-QAM signal sets the  Latin Square given by bitwise Exclusive-OR (XOR) is not a standard Latin Square for any $M>4.$ This can be easily seen as follows: Any square $M$-QAM signal set ($M>4$) has points of the form $a+jc, a+j(c+b), a+j(c-b)$, for some integers $a$, $b$ and $c$. For $z=1$, the effective constellation at R during the MA phase contains the point $2(a+jc)$ and can be resulted in at least two different ways, since $2(a+jc)=(a+jc)+z(a+jc)= a+j(c+b)+z(a+j(c-b))$ for $z=1$. Let $ l_1, l_2$ and $l_3$ denote the labels for $a+jc, a+j(c+b)$, and $a+j(c-b)$ respectively. For the singular fade state $z=1,$ we have $\lbrace(l_1,l_1),(l_2,l_3) \rbrace$ as a singularity removal constraint. But the Latin Square obtained by bitwise XOR mapping does not satisfy this constraint since $l_1\oplus l_1=0 \neq l_2\oplus l_3$.

% \vspace{-.3 in}
\subsection{Standard Latin Square for $\sqrt M-$PAM}
In this subsection, we obtain standard Latin Squares for $\sqrt M-$PAM signal sets. 
%%% 
\begin{definition}
An $M \times M$ Latin square in which each row is obtained by a left cyclic shift of the previous row is called a left-cyclic Latin Square.
\end{definition}
%%%%%
\begin{lemma}
\label{pam_std_lemma}
For a $\sqrt M$-PAM signal set a left-cyclic Latin Square removes the singular fade state $z=1$.  
\end{lemma}
%%%%%
\begin{proof}
See Appendix \ref{app8}.
\end{proof}
%%%%%%%%%%
\begin{example}
Consider the received constellation at the relay when the end nodes use 4-PAM constellation and let the channel condition be $z=1$ as given in Fig.\ref{fig:pam_received}.
The singularity removal constraints are 
\begin{align*}
\{(0,1)(1,0)\}&,~\{(0,2)(1,1)(2,0)\},~\{(0,3)(1,2)(2,1)(3,0)\},\\ \left.
\right.
&\{(1,3)(2,2)(3,1)\},~ \mbox{and} ~\{(2,3)(3,2)\}.
\end{align*}
\noindent
The Latin Square which removes this singular fade state is given in Fig.\ref{cyclic_LS}.

\begin{figure}[h]
\centering
%\vspace{-.8 cm}
\includegraphics[totalheight=.3in,width=2.5in]{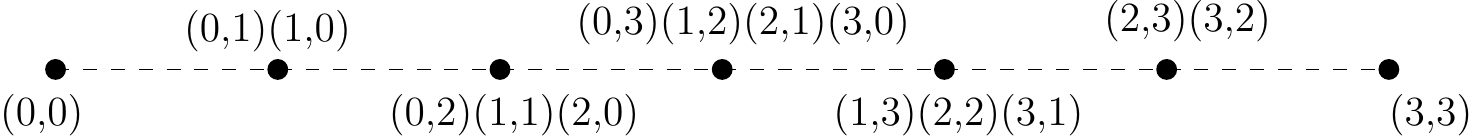}
%\vspace{-2 cm}
\caption{Received Constellation at the relay for $z=1$.}     
\label{fig:pam_received}
%\vspace{-.5 cm}        
\end{figure}
%%%%%%%%%%%%%%%%%%%%%%%%%%%%%%%%%%%%%%%%%%%%%%%%%%%%%%%%%%%%%%%
\begin{figure}[h]
\centering
%\vspace{-.5 cm}
\begin{tabular}{|c|c|c|c|}
\hline 0 & 1 & 2 & 3 \\ 
\hline 1 & 2 & 3 & 0 \\ 
\hline 2 & 3 & 0 & 1\\
\hline 3 & 0 & 1 & 2 \\
\hline 
\end{tabular}
\caption{Left-cyclic Latin Square to remove the singular fade state $z=1$}
%\vspace{-1 cm}
\label{cyclic_LS}
\end{figure}
\end{example}
%%%%%%%%%%%%%%%%%%%%%%%%%%%%%%%%%%%%%%%%%%%%%%%%%%%%%%%%%

%%%%%%%%%%%%%%%%%%%%%%%%%%%%%%%%%%%%%%%%%%%%%%%%%%%%%%%%%%%%%%
\subsection{Standard Latin Square for $M-$QAM}
In this subsection standard Latin Square for a square $M$-QAM constellation is obtained from that of $\sqrt M$-PAM constellation. 

Let $PAM(i),$ for $i=1,2,\cdots, {\sqrt M},$ denote the symbol set consisting of $\sqrt M$ symbols $\{(i-1)\sqrt M, ((i-1)\sqrt M)+1, ((i-1)\sqrt M)+2,\cdots,((i-1)\sqrt M)+(\sqrt M-1)\}.$ Let $L_{PAM(i)}$ denote the standard Latin Square for $\sqrt M$-PAM, with  symbol set $PAM(i)$ and also let  $L_{QAM}$ denote the standard Latin Square for $M$-QAM. Then, $L_{QAM}$ is given in terms of $L_{PAM(i)},$ $i=1,2,\cdots, \sqrt M,$ as the block left-cyclic Latin Square shown in Fig. \ref{fig:qam_const}. This is formally shown in the following Lemma.
%%%%%%%%%%%%%%%%%%%%%%%%

\begin{lemma}
\label{qam_ls}
Let $PAM(i)$  for $i=1,2,\cdots, {\sqrt M},$ denote the symbol set consisting  of $\sqrt M$ symbols $\{(i-1)\sqrt M, ((i-1)\sqrt M)+1, ((i-1)\sqrt M)+2,\cdots,((i-1)\sqrt M)+(\sqrt M-1)\}$ and let  $L_{PAM(i)}$ stand for the Latin Square that removes the singular fade state $z=1$ with the symbol set $PAM(i)$ for $\sqrt M$-PAM. Then arranging the cyclic Latin Squares $L_{PAM(i)}$ as shown Fig.\ref{fig:qam_const} where each row is a block wise left-cyclically shifted version of the previous row results in a Latin Square which removes the singular fade state $z=1$ for $M$-QAM. 
\end{lemma}
%%%%%%%%%%%%%%%%%%%
\begin{proof}
See Appendix \ref{app9}.
\end{proof}
%%%%%%%%%%%%%
\begin{figure}[h]
\centering
%\vspace{-.6 cm}
\includegraphics[totalheight=2.2in,width=2.2in]{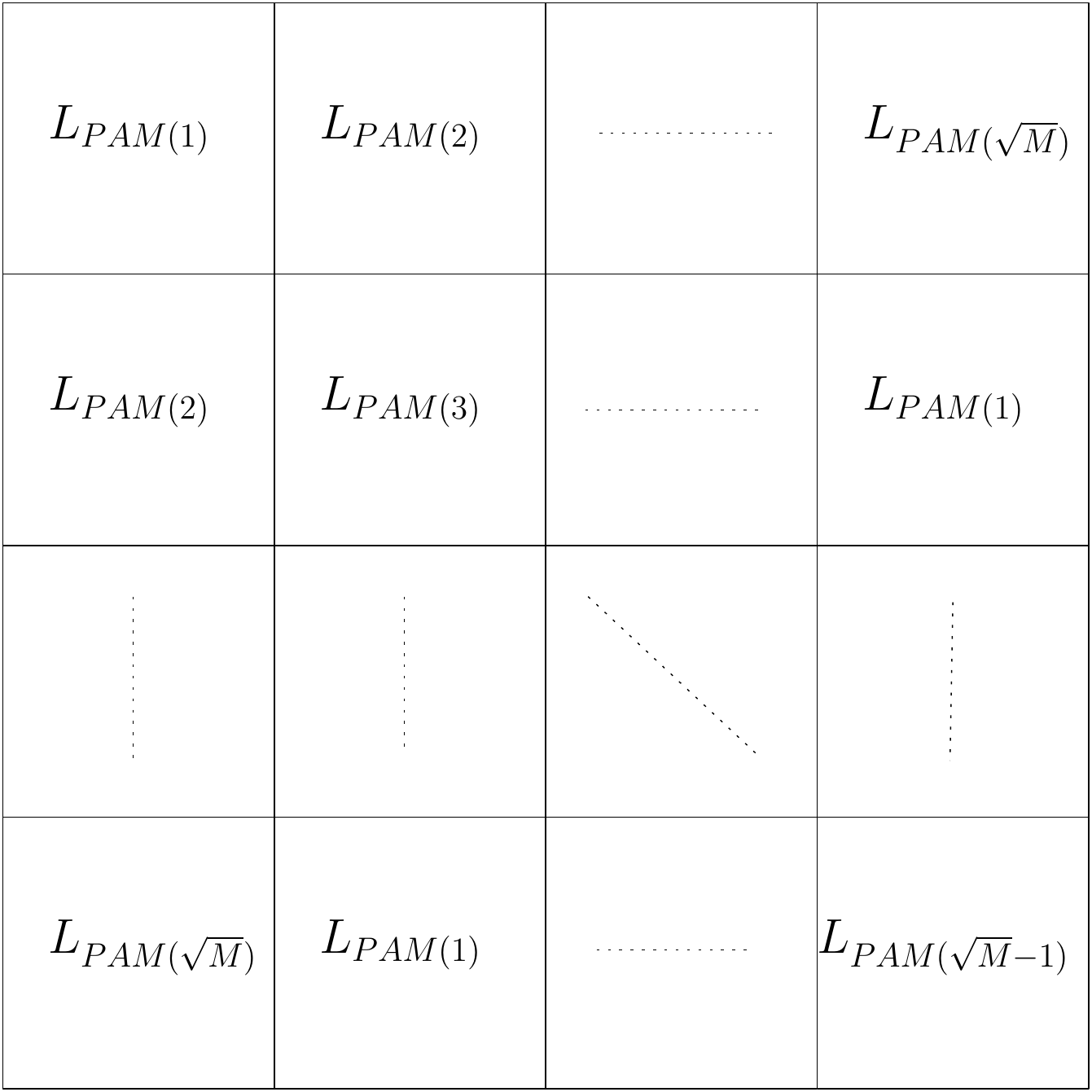}
%\vspace{-2 cm}
\caption{Construction of $L_{QAM}$ for $z=1$.}     
\label{fig:qam_const}  
%\vspace{-1 cm}      
\end{figure}
%%%%%%%
The standard Latin Square for 16-QAM is shown in Fig.\ref{qam_lat}.
\begin{figure*}[htbp]
\centering
%\vspace{-.5 cm}
\begin{tabular}{|c|c|c|c||c|c|c|c||c|c|c|c||c|c|c|c|}
\hline 0 & 1 & 2 & 3 & 4 & 5 & 6 & 7 & 8 & 9 & 10 & 11 & 12 & 13 & 14 & 15\\
\hline 1 & 2 & 3 & 0 & 5 & 6 & 7 & 4 & 9 & 10 & 11 & 8 & 13 & 14 & 15 & 12\\ 
\hline 2 & 3 & 0 & 1 & 6 & 7 & 4 & 5 & 10 & 11 & 8 & 9 & 14 & 15 & 12 & 13\\
\hline 3 & 0 & 1 & 2 & 7 & 4 & 5 & 6 & 11 & 8 & 9 & 10 & 15 & 12 & 13 & 14\\
\hline
\hline 4 & 5 & 6 & 7 & 8 & 9 & 10 & 11 & 12 & 13 & 14 & 15 & 0 & 1 & 2 & 3\\
\hline 5 & 6 & 7 & 4 & 9 & 10 & 11 & 8 & 13 & 14 & 15 & 12 & 1 & 2 & 3 & 0\\ 
\hline 6 & 7 & 4 & 5 & 10 & 11 & 8 & 9 & 14 & 15 & 12 & 13 & 2 & 3 & 0 & 1\\
\hline 7 & 4 & 5 & 6 & 11 & 8 & 9 & 10 & 15 & 12 & 13 & 14 & 3 & 0 & 1 & 2\\
\hline
\hline 8 & 9 & 10 & 11 & 12 & 13 & 14 & 15 & 0 & 1 & 2 & 3 & 4 & 5 & 6 & 7\\
\hline 9 & 10 & 11 & 8 & 13 & 14 & 15 & 12 & 1 & 2 & 3 & 0 & 5 & 6 & 7 & 4\\ 
\hline 10 & 11 & 8 & 9 & 14 & 15 & 12 & 13 & 2 & 3 & 0 & 1 & 6 & 7 & 4 & 5\\
\hline 11 & 8 & 9 & 10 & 15 & 12 & 13 & 14 & 3 & 0 & 1 & 2 & 7 & 4 & 5 & 6\\
\hline
\hline 12 & 13 & 14 & 15 & 0 & 1 & 2 & 3 & 4 & 5 & 6 & 7 & 8 & 9 & 10 & 11\\
\hline 13 & 14 & 15 & 12 & 1 & 2 & 3 & 0 & 5 & 6 & 7 & 4 & 9 & 10 & 11 & 8\\ 
\hline 14 & 15 & 12 & 13 & 2 & 3 & 0 & 1 & 6 & 7 & 4 & 5 & 10 & 11 & 8 & 9\\
\hline 15 & 12 & 13 & 14 & 3 & 0 & 1 & 2 & 7 & 4 & 5 & 6 & 11 & 8 & 9 & 10\\
\hline 
\end{tabular}
\caption{Standard Latin Square $L_{QAM}$ for 16-QAM.}
%\vspace{-1 cm}
\label{qam_lat}
\end{figure*}
%%%%%%%%%%%%%%%%%%%%%%%%%%%%%%%%%%%%%%%%%%%%%%%%%%%%%%%%%%%%%%
%%%%%%%%%%%%%%%%%%%%%%%%%%%%%%%%%%%%%%%%%%%%%%%%%%%%%%%%%%%%%%%%%%%%%%%%%%%%%%%%%%%%%%%%%%%
\section{Channel quantization for $M$-QAM signal sets}
\label{sec4}
In Section \ref{sec3}, we saw how the relay R chooses a complex number for transmission during the broadcast phase if the channel fade coefficients result in a singular fade state. This was done essentially by associating a Latin Square to each singular fade state. However, ${\gamma}e^{j\theta} $ being a ratio of two complex numbers ($ h_B/ h_A $), can take any value in the complex plane (in fact, it takes a value equal to a singular fade state with probability zero). This demands the partition of the complex plane into regions and associating a Latin Square to each region so as to optimize total number of network coding maps used at the relay during the broadcast phase. The study of channel quantization for $M$-PSK signal set for two-user fading MAC channel is described in \cite{KR} and that for bi-directional relay network is addressed in \cite{VNR}. While \cite{KR} uses the notion of distance classes for partitioning the complex plane, in \cite{VNR} the regions associated with each singular fade state are obtained by plotting the pair-wise transition boundaries of relevant pairs of singular fade states. 

In this section we describe how to partition the set of all possible ${\gamma}e^{j\theta} $ when the end nods A and B use a square $M$-QAM signal set. The salient difference while using $M$-QAM constellation when compared to $M$-PSK signal sets at nodes A and B is as follows. It is shown in Lemma 1 of \cite{VNRarX} that no two distinct pairs  of points $(d_{k_1},d_{l_1})$ and $(d_{k_2},d_{l_2})$ in the difference constellation $\Delta {\mathcal{S}}$, such that $|d_{k_1}| \neq |d_{k_2}|$ and $|d_{l_1}| \neq |d_{l_2}|$, can have the same ratio when $\mathcal{S}$ is $M$-PSK signal set, i.e.,
\begin{align}
\label{diffwithpsk}
z=-\frac{d_{k_1}}{d_{l_1}} = -\frac{d_{k_2}}{d_{l_2}},~~ d_{k_1}\neq d_{l_1},d_{k_2}\neq d_{l_2} \in \Delta {\mathcal{S}}\\
\nonumber
\mathrm{if~ and~ only~ if}~ |d_{k_1}|=|d_{k_2}| ~\mathrm{and}~ |d_{l_1}|=|d_{l_2}|.
\end{align}

In other words, all $z \neq z' \neq 1 \in {\mathcal{H}}$, arise from distinct pairs differing in their magnitudes in $\Delta {\mathcal{S}}$ for $M$-PSK. This is not the case when the end nodes A and B use $M$-QAM signal set. This is shown in the following example. First, notice that for $M$-QAM signal set used at nodes A and B, any point $d_k \in \Delta \cal S$ can be written in the form
\begin{align}
\nonumber
\label{eqndifc}
d_k=\frac{2}{\sqrt{\rho}} \left( n_k + j m_k\right), \mathrm{for~some} ~n_k, m_k \in \left\lbrace -\left(\sqrt{M}-1 \right), \right. \\ \left. -\left(\sqrt{M}-2 \right), ..., \left(\sqrt{M}-2 \right), \left(\sqrt{M}-1 \right)\right\rbrace
\end{align}
\begin{example}\label{egdiff}
Consider 16-QAM signal set used at nodes A and B. For the following 4 points in $\Delta \mathcal{S}$, $d_{k_1}=\frac{2}{\sqrt{\rho}} (2+j)$ $\neq$ $d_{l_1}=\frac{2}{\sqrt{\rho}} (1+0j)$, $d_{k_2}=\frac{2}{\sqrt{\rho}} (1+3j)$ $\neq$ $d_{l_2}=\frac{2}{\sqrt{\rho}} (1+j)$,
\begin{align*}
1 \neq \frac{d_{k_1}}{d_{l_1}} =  \frac{d_{k_2}}{d_{l_2}}
\end{align*}
However, $|d_{k_1}| \neq |d_{k_2}|$ and $|d_{l_1}| \neq |d_{l_2}|$.
\end{example}

So the procedure used for channel quantization in \cite{VNR} for $M$-PSK does not directly apply for $M$-QAM. However, as in case with $M$-PSK described in \cite{VNR}, we show that for certain values of ${\gamma}e^{j\theta} $ any choice of clustering satisfying the exclusive law gives the same minimum cluster distance so that any one of the Latin Squares may be chosen. Subsequently, channel quantization for those values of ${\gamma}e^{j\theta} $ for which the choice of the Latin Square does play a role in the overall performance is taken up. 

\subsection{Clustering Independent Region}
%\vspace{-.1 in}
\begin{definition}
The set of values of ${\gamma}e^{j\theta} $ for which any clustering satisfying the exclusive law gives the same minimum cluster distance is referred to as the clustering independent region. The region in the complex plane other than the clustering independent region is called the clustering dependent region.
\end{definition}
%%%%%%%%%%%%%%%%%%%%%%%%%%%%%%%%%%%%%%%%%%%%%%%%%%%%%%%%%%%%%%%%%%%%%%

It is shown in \cite{KR} that with an arbitrary signal set $\cal S$ used at A and B for the MA phase, for any clustering $ \cal C^{\gamma,\theta} $ satisfying the exclusive law, we have 
\begin{align}
\label{eqnset1}
&d_{min}({\cal C}^{\gamma,\theta}) \leq \min \{ d_{min}({\cal S}), \gamma d_{min}({\cal S}) \}
\end{align}

%%%%%%%%%%%%%%%%%%%%%%%%%%%%%%%%%%%%%%%%%%%%%%%%%%%%%%%%%%%%%%%%%%%%%%
%%%%%%%%%%%%%%%%%%%%%%%%%%%%%%%%%%%%%%%%%%%%%%%%%%%%%%%%%%
\begin{observation}
\label{obs1}
Using \eqref{eqnset1}, for values of $\gamma e^{j\theta}$ for which ${d_{min}(\gamma e^{j\theta}) \geq \min \{ d_{min}({\cal S}), \gamma d_{min}({\cal S})\}}$, since ${d_{min}(\gamma e^{j\theta}) \leq d_{min}({\cal C}^{\gamma,\theta})}$, we have ${d_{min}({\cal C}^{\gamma,\theta}) = \min \{ d_{min}({\cal S}), \gamma d_{min}({\cal S}) \}}$ for all clusterings ${\cal C}^{\gamma,\theta}$ satisfying the exclusive law. Such $\gamma e^{j\theta}$ therefore belong to the clustering independent region.
\end{observation}Define
\begin{align}
\nonumber
\label{eqnGCI}
\vspace{-.1 in}
\Gamma_{CI}({\cal S}) = &\{ \gamma e^{j\theta} : \left| d_k + \gamma e^{j\theta} d_l \right| \geq \min	( d_{min}({\cal S}), \gamma d_{min}({\cal S})) \\
&\forall \left( d_k, d_l \right)\neq \left( 0, 0 \right) \in (\Delta {\cal S})^2,  \gamma \in {\mathbb{R}}^{+}, - \pi \leq \theta < \pi \}\\
&= \Gamma_{CI}^{ext}({\cal S}) ~{\cup} ~\Gamma_{CI}^{int}({\cal S})
\end{align}
where 
%\vspace{-.1 in}
\begin{align}
\label{eqnGCIext}
\Gamma_{CI}^{ext}({\cal S})= \Gamma_{CI}({\cal S})\cap \left\lbrace \gamma > 1\right\rbrace,
\end{align}
%\vspace{-.4 in}
\begin{align}
\Gamma_{CI}^{int}({\cal S})= \Gamma_{CI}({\cal S})\cap \left\lbrace \gamma \leq 1\right\rbrace.
\end{align}
It has been shown in \cite{VNR} that the region $\Gamma_{CI}^{int}({\cal S})$ is obtained by the complex inversion of the region $\Gamma_{CI}^{ext}({\cal S})$ so that by finding out one the other can be obtained easily.
\begin{theorem}
\label{Thrm1}
For $M$-QAM signal set the region $\Gamma_{CI}^{ext}(\mbox{$M$-QAM})$ is the outer envelope region formed by $8\left(\sqrt{M}-1 \right)$ unit circles with centers ($\alpha +j\beta $) belonging to the set $\left\lbrace \pm\left(\sqrt{M}-1 \right)+ jx, x \pm j\left(\sqrt{M}-1 \right)\right\rbrace$, where $x \in {\left\lbrace {-\left(\sqrt{M}-1 \right)}, {-\left(\sqrt{M}-2 \right)},..., {\left(\sqrt{M}-2 \right)}, {\left(\sqrt{M}-1 \right)} \right\rbrace} $
\end{theorem}
\vspace{0.2 cm}
\begin{proof}
See Appendix $\ref{appen1}$
\end{proof}
It can be verified that the centers of the $8\left(\sqrt{M}-1 \right)$ circles in Theorem \ref{Thrm1} are the singular fade states which lie on the outermost square.

For $M>4$ the region $\Gamma_{CI}^{int}({\cal S})$ is described in the following lemma. Here, it may be noted that, for normalized signal sets used at nodes A and B, $4$-PSK and $4$-QAM are the same, and $M$-PSK (for all $M$) has already been addressed in \cite{VNR}.

%%%%%%%%%%%%%%%%%%%%%%%%%%%%%%%%%%%%%%%%%%%%%%%%%%%%%%%%%%%%%%%%
\begin{lemma}
\label{Lemma2}
For $M>4$, $\Gamma_{CI}^{int}(\mbox{$M$-QAM})$ is the outer envelope region formed by the \mbox{$8(\sqrt{M}-1)$} circles with centers \mbox{$\left(\frac{\alpha}{{\alpha}^2+{\beta}^2-1} -j \frac{\beta}{{\alpha}^2+{\beta}^2-1}\right)$}  and radii $1/\left(\alpha^2+\beta^2-1\right)$ where \mbox{$(\alpha +j\beta )$} belongs to the set $\left\lbrace \pm\left(\sqrt{M}-1 \right)+ jx, x \pm j\left(\sqrt{M}-1 \right)\right\rbrace$, $x \in {\left\lbrace {-\left(\sqrt{M}-1 \right)},{-\left(\sqrt{M}-2 \right)},..., {\left(\sqrt{M}-2 \right)}, {\left(\sqrt{M}-1 \right)} \right\rbrace} $
\end{lemma}
\begin{proof}
Proof follows directly from the complex inversion of the $8\left(\sqrt{M}-1 \right)$ circles used for obtaining $\Gamma_{CI}^{ext}({M-QAM})$ in Theorem $\ref{Thrm1}$.
\end{proof}
%%%%%%%%%%%%%%%%%%%%%%%%%%%%%%%%%%%%%%%%%%%%%%%%%%%%%%%%%%%%%%%%
\vspace{0.1 in}
Define
\begin{align}
\label{infinitynorm}
||\gamma e^{j\theta}||_{\infty} \triangleq \max \left\{ |\Re (\gamma e^{j\theta})|, |\Im (\gamma e^{j\theta})|\right\},~\forall \gamma e^{j\theta}\in {\mathbb{C}}
\end{align}
where $\Re (z)$ and $\Im (z)$ stand for the real and imaginary parts respectively of the complex number $z$.
%%%%%%%%%%%%%%%%%%%%%%%%%%%%%%%%%%%%%%%%%%%%%%%%%%%%%%%%%%%%%%%%%
%\begin{figure}[htbp]
%\centering
%\vspace{-.5 in}
%\includegraphics[totalheight=2.6in,width=3.5in]{extCI_16.png} 
%\caption{Diagram showing $\Gamma_{CI}^{ext}$ (shaded region) for 16-QAM signal set}	
%\label{fig:extCI_16}
%% \vspace{-1cm}
%\end{figure}

%%%%%%%%%%%%%%%%%%%%%%%%%%%%%%%%%%%%%%%%%%%%%%%%%%%%%%%%%%%%%%%%%
From Theorem $\ref{Thrm1}$, it can be noted that for $||\gamma e^{j\theta}||_{\infty} \geq \sqrt{M}$, $\gamma e^{j\theta}$ belongs to $\Gamma_{CI}^{ext}$ and for $|\gamma e^{j\theta}| \leq \frac{1}{\sqrt{2M}+1-\sqrt{2}}$,  $\gamma e^{j\theta}$ belongs to $\Gamma_{CI}^{int}$. In the complex plane, the locus of the points satisfying $||\gamma e^{j\theta}||_\infty=a$ is a square centered at origin and having sides of length $a$. So, if $\gamma e^{j\theta}$ lies outside the square centered at origin and having sides of length $\sqrt{M}$ or inside the circle centered at the origin with radius $\frac{1}{\sqrt{2M}+1-\sqrt{2}}$ the relay can choose to use a fixed predetermined clustering satisfying the mutual exclusive law. In particular, that clustering whose corresponding Latin Square has only $M$ symbols can be used. This observation helps in significantly reducing the computational complexity at the relay. The relay needs to adaptively switch between network coding maps only if the above two conditions are not satisfied. The following example illustrates Theorem $\ref{Thrm1}$ and the above observation.
%%%%%%%%%%%%%%%%%%%%%%%%%%%%%%%%%%%%%%%%%%%%%%%%%%%%%%%%%%%%%%%%%
\begin{example}
Consider the case when the nodes A and B use $16$-QAM signal set. For this scenario, according to Theorem $\ref{Thrm1}$, the 24 unit circles centered at the singular states lying in the outermost square are shown in Fig. $\ref{fig:extCI_16}$. The region $\Gamma_{CI}^{ext}$ is the shaded region in Fig. $\ref{fig:extCI_16}$. The region $\Gamma_{CI}^{int}$ is the shaded region in Fig. $\ref{fig:intCI_16}$. The circles in Fig. $\ref{fig:intCI_16}$ are those obtained by the complex inversion of the unit circles shown in Fig. $\ref{fig:extCI_16}$. The clustering independent region $\Gamma_{CI}$ is the union of the shaded regions in Fig. $\ref{fig:extCI_16}$ and Fig. $\ref{fig:intCI_16}$. Additionally, as an approximation, if $\gamma e^{j\theta}$ falls in the region between the square with side length 4 (B1 in Fig. $\ref{fig:CIregion_16}$)  and circle with radius $\frac{1}{3\sqrt{2}+1}$ (B2 in Fig. $\ref{fig:CIregion_16}$), the relay uses adaptive network coding maps. This is shown in Fig. $\ref{fig:CIregion_16}$.
\end{example}
%%%%%%%%%%%%%%%%%%%%%%%%%%%%%%%%%%%%%%%%%%%%%%%%%%%%%%%%%%
\begin{figure}[htbp]
\centering
%\vspace{-0.4 in}
\subfigure[$\Gamma_{CI}^{ext}$ (shaded region) for 16-QAM]{
\includegraphics[totalheight=2.4in,width=3in]{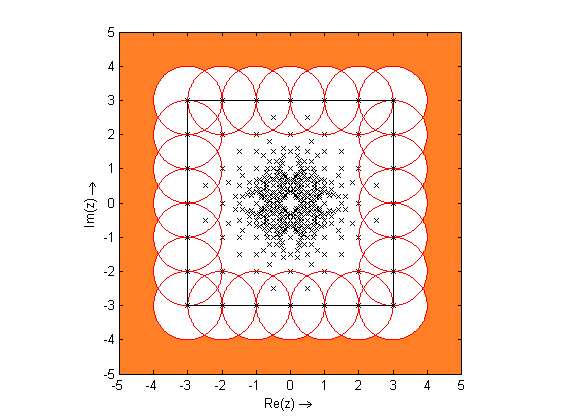}
\label{fig:extCI_16}
}
%\vspace{-0.5 in}
\subfigure[$\Gamma_{CI}^{int}$ (shaded region) for 16-QAM]{
\includegraphics[totalheight=2.4in,width=3in]{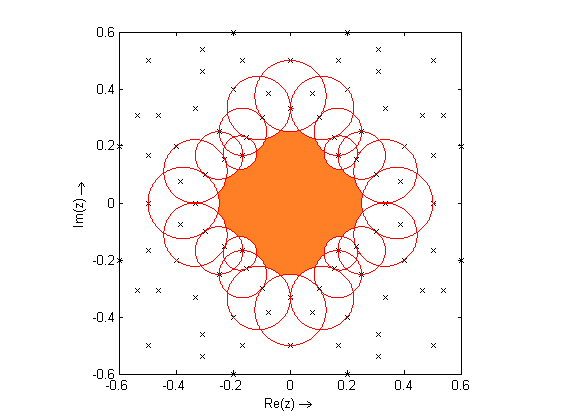}
\label{fig:intCI_16}
}
\label{Cluster_indep}

\caption{Cluster Independent regions for 16 QAM}
\end{figure}
%%%%%%%%%%%%%%%%%%%%%%%%%%%%%%%%%%%%%%%%%%%%%%%%%%%%%5

%%%%%%%%%%%%%%%%%%%%%%%%%%%%%%%%%%%%%%%%%%%%%%%%%%%%%%%%%%%%%%%%%
%
%\begin{figure}[htbp]
%\centering
%\includegraphics[totalheight=2.5in,width=3.5in]{intCI_16.png} 
%\caption{Diagram showing $\Gamma_{CI}^{int}$ (shaded region) for 16-QAM signal set}	
%\label{fig:intCI_16}
%% \vspace{-1cm}
%\end{figure}

\begin{figure}[htbp]
\centering
\includegraphics[totalheight=2.5in,width=3.5in]{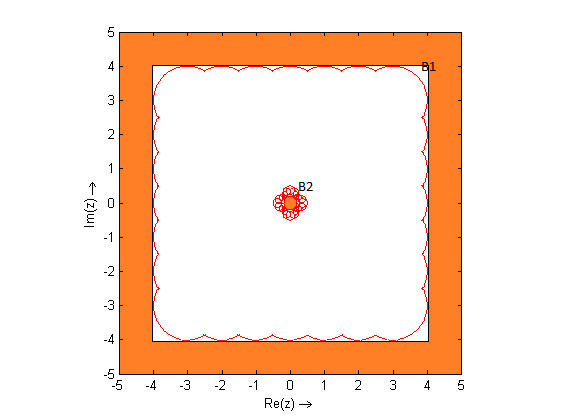} 
\caption{The unshaded region shows the region in the complex plane where adaptive switching needs to be done for 16-QAM signal set}	
\label{fig:CIregion_16}
\end{figure}

%%%%%%%%%%%%%%%%%%%%%%%%%%%%%%%%%%%%%%%%%%%%%%%%%%%%%%%%%%%%%%%%%%%

\subsection{Clustering Dependent Region}
In this section we partition that region of the complex plane where the choice of Latin Square significantly determines the performance of the bi-directional relay during the broadcast phase. It follows from Observation $\ref{obs1}$ that values of $\gamma e^{j\theta}$ for which $d_{min}(\gamma e^{j\theta}) < \min \{ d_{min}({\cal S}), \gamma d_{min}({\cal S}) \}$ constitute the clustering dependent region.
As the criterion described in \cite{VNR} for partitioning the channel fade state complex plane when nodes A and B use $M$-PSK depends on (\ref{diffwithpsk}) which no longer holds for $M$-QAM as illustrated in Example (\ref{egdiff}), we develop an alternative criterion for performing the channel partitioning in the clustering dependent region.\\
With $ {\cal D} \left( \gamma, \theta, d_k, d_l \right)$ defined as
\begin{align}
{\cal D} \left( \gamma, \theta, d_k, d_l \right)= |d_k +\gamma e^{j\theta}d_l|,~\left( d_k, d_l\right) \neq \left( 0,0 \right) \in (\Delta {\cal S})^2,
\end{align}
we have the following lemma.
%%%%%%%%%%%%%%%%%%%%%%%%%%%%%%%%%%%%%%
\begin{lemma}
\label{Lemma3}
If $\gamma e^{j\theta}$ is such that $arg \max_{z \in {\cal H}} \min_{\left( d_{k}, d_{l} \right):\frac{-d_k}{d_l}=z} {\cal D} \left( \gamma, \theta, d_{k}, d_{l} \right) = z'$, then the clustering ${\cal C}^{\left \{z' \right \}}$ maximizes the minimum cluster distance, among all the clusterings which belong to the set ${\cal C}_{\cal H}$
\end{lemma}
\begin{proof}
For $d_k$ and $d_l$ $\in \Delta {\cal S}$
\begin{align}
\label{eqn1appen2}
\min_{\left( d_{k}, d_{l} \right):\frac{-d_k}{d_l}=z} {\cal D} \left( \gamma, \theta, d_{k}, d_{l} \right)= d_{min}({\cal C}^{\left \{ z\right \}},\gamma,\theta),
\end{align}
the minimum cluster distance of the clustering ${\cal C}^{\left \{ z\right \}}$ evaluated at $\gamma e^{j\theta}$. If the maximum in ($\ref{eqn1appen2}$) is achieved for $z' \in {\cal H}$, then $d_{min}({\cal C}^{\left \{ z'\right \}},\gamma,\theta)$ $\geq$ $d_{min}({\cal C}^{\left \{ z\right \}},\gamma,\theta)$ $\forall z \neq z' \in {\cal H}$. Then  ${\cal C}^{\left \{ z'\right \}}$ is the best clustering for the channel fade coefficient $\gamma e^{j\theta}$.
\end{proof}
%%%%%%%%%%%%%%%%%%%%%%%%%%%%%%%%%%%%%%%%%%%%%%%%%%%%%%%%%%%%%%%%%%%%5
%\vspace{0.5 cm}

Therefore, associated with each singular fade state $z \in {\cal H}$ we have a region ${\cal R}_{\left\lbrace z\right\rbrace}$ in the $\gamma e^{j\theta}$ plane in which the clustering  ${\cal C}^{\left \{z \right \}}$ maximizes the minimum cluster distance. This region  ${\cal R}_{\left\lbrace z\right\rbrace}$ is given by
\begin{align*}
{\cal R}_{\left\lbrace z\right\rbrace}&=\{ \gamma e^{j\theta} : \hspace{-0.4 cm}\min_{(d_k,d_l)  \in (\Delta {\cal S})^2: \frac{-d_k}{d_l}=z} |d_k + \gamma e^{j\theta} d_l|~ \geq \\
&~~~~~~\min_{(d_k,d_l)  \in (\Delta {\cal S})^2: \frac{-d_k}{d_l}=z'} |d_{k} + \gamma e^{j\theta} d_{l}|,\forall z' \neq z \in {\cal H}\}.
\end{align*}

The boundaries of these regions for each singular fade state are explicitly derived next. It is shown that like with $M$-PSK signal set considered in \cite{NVR}, the boundaries of the region ${\cal R}_{\left\lbrace z \right\rbrace}$ are either circles or straight lines and a systematic procedure to obtain these regions for each singular state is given. A simple formulation to find out the pair wise transition boundary corresponding to a pair of clusterings ${\cal C}^{\left\lbrace z_1 \right\rbrace}$ and ${\cal C}^{\left\lbrace z_2 \right\rbrace}$ is stated next.

The curve $c( z_1, z_2)$ which denotes the pair-wise transition boundary formed by the singular fade state $z_1$ with the singular fade state $z_2$ is the set of $\gamma e^{j\theta}$ for which
\begin{align}
\nonumber
\min_{\substack {{(d_k,d_l) \in (\Delta {\cal S})^2 :} \\ {\frac{-d_k}{d_l}=z_1}}} { \vert {d_k+\gamma e^{j\theta} d_l} \vert}  = \min_{\substack {{(d_k,d_l) \in (\Delta {\cal S})^2 :} \\ {\frac{-d_k}{d_l}=z_2}}} { \vert {d_k+\gamma e^{j\theta} d_l} \vert}.
\end{align}

\begin{theorem}
\label{Thrm2}
With the notations
\begin{align*}
\check{d_l}&={\mathrm{arg}} \hspace{-0.8 cm}\min_{d_2 \in \Delta {\cal S} ~:~ -\frac{d_1}{d_2}=z_1} \left\{ \vert d_1 + \gamma e^{j\theta}d_2 \vert \right\},\\
\check{d_{l'}}&={\mathrm{arg}} \hspace{-0.8 cm}\min_{d_2 \in \Delta {\cal S} ~:~ -\frac{d_1}{d_2}=z_2} \left\{ \vert d_1 + \gamma e^{j\theta}d_2 \vert \right\},
\end{align*}
the pair wise transition curve $c\left( z_1, z_2\right)$, $z_1\neq z_2$ is any one of the following
\begin{itemize}
\item if $|\check{d_l}| \neq |\check{d_{l'}}|$, a circle with center $\left( x, y\right)$ and radius $r$, where
\begin{align*}
&x=\frac{\Re\left( z_1\right)}{1- {| \frac{\check{d_{l'}}}{\check{d_l}} |}^2} + \frac{\Re\left( z_2\right)}{1- {| \frac{\check{d_l}}{\check{d_{l'}}} |}^2},~
y=\frac{\Im\left( z_1\right)}{1- {| \frac{\check{d_{l'}}}{\check{d_l}} |}^2} + \frac{\Im\left( z_2\right)}{1- {| \frac{\check{d_l}}{\check{d_{l'}}} |}^2}\\
&\mathrm{and}~r=\sqrt{\left( x^2 +y^2 + \frac{|z_2|^2 |\check{d_{l'}}|^2-|z_1|^2 |\check{d_{l}}|^2}{|\check{d_{l}}|^2-|\check{d_{l'}}|^2}\right)}
\end{align*}

\item {if $|\check{d_l}| = |\check{d_{l'}}|$, a straight line of the form $ax + by=c$, where \vspace{-0.1 in}\begin{align*} 
a&=\left( {\Re\left( z_1\right)|\check{d_l}|^2-\Re\left( z_2\right)|\check{d_{l'}}|^2} \right),\\
b&=\left( {\Im\left( z_1\right)|\check{d_l}|^2-\Im\left( z_2\right)|\check{d_{l'}}|^2} \right),\\
c&=-\frac{1}{2} \left( {|z_2|^2 |\check{d_{l'}}|^2-|z_1|^2 |\check{d_{l}}|^2} \right).
\end{align*}}
\end{itemize}
\end{theorem}
\begin{proof}
See Appendix $\ref{appen3}$
\end{proof}

\begin{observation}
\label{obs2}
From \cite{VNR} it is known that the region ${\cal R}_{\left\{-d_l/d_k\right\}}$ is the region obtained by the complex inversion of the region ${\cal R}_{\left\{-d_k/d_l\right\}}$. Also the distribution of the singular fade state in the $\gamma e^{j\theta}$ complex plane is periodic with periodicity $\pi/2$ for $M$-QAM. Moreover, within an interval $\left[ a,a+\pi/2\right], a\in \left\{ 0, \pi/2, \pi, 3\pi/2 \right\}$, the singular fade states are symmetric with respect to the  line $\theta = a+\pi/4$.
\end{observation}

\subsubsection*{Channel Quantization of 16-QAM Signal Set}
We consider the case when nodes A and B both use 16-QAM signal set. From Observation $\ref{obs2}$, it follows that we need to consider only those singular fade states which lie outside the unit circle and within the angular interval $\theta \in \left[ 0, \pi/4\right]$. The pair wise transition boundaries as given by Theorem $\ref{Thrm2}$ for adjacent pairs of singular fade states in this region give the region ${\cal R}_{\left\{-d_k/d_l\right\}}$, $|d_k| > |d_l|$ and $\measuredangle{ (-d_k/d_l) } \in \left[ 0, \pi/4\right]$. Using the symmetry property and periodicity as mentioned in Observation $\ref{obs2}$ we get the region ${\cal R}_{\left\{-d_k/d_l\right\}}$ corresponding to all singular fade states lying outside the unit circle centered at the origin. Again, from Observation $\ref{obs2}$, the regions corresponding to the singular fade states lying inside the unit circle is obtained by complex inversion of those obtained for singular fade states lying outside the unit circle. The regions corresponding to singular fade states lying on the unit circle are the remaining regions in the complex plane after the regions corresponding to all the other singular fade states are obtained.
%%%%%%%%%%%%%%%%%%%%%%%%%%%%%%%%%%%%%%%%%

%\vspace{-1 cm}
%\begin{figure}[htbp]
%\centering
%\includegraphics[totalheight=2.5in,width=3.5in]{onesfs.png} 
%\caption{Diagram showing the pairwise transition boundaries corresponding to the singular fade state $2+1j$}
%\label{fig:onesfs}
%% \vspace{-1cm}
%\end{figure}
%
%\begin{figure}[htbp]
%\centering
%\includegraphics[totalheight=2.5in,width=3.5in]{2_1jsfs.png} 
%\caption{The shaded region shows the region ${\cal R}_{\left\{ 2+1j\right\}}$ for 16-QAM signal set}
%\label{fig:2_1jsfs}
%% \vspace{-1cm}
%\end{figure}
%%%%%%%%%%%%%%%%%%%%%%%%%%%%%%%%%%%%%%%%%%%%
For 16-QAM signal set used at nodes A and B, there are 27 singular fade states in the region outside the unit circle and in the angular interval $\theta \in \left[ 0, \pi/4\right]$. To illustrate Theorem $\ref{Thrm2}$ consider the singular fade state $2+j$ in the said interval. It shares pairwise transition boundaries with 8 neighbouring singular fade states viz. $\{1.5+1.5j,~1.5+j,~1.4+0.8j,~1.8+0.6j,~2+2j,~3+j, 2.5+0.5j,~2 \}$. The pairwise transition boundary between $2+j$ and $1.5+1.5j$ is the circle $C1$ shown in Fig. $\ref{fig:onesfs}$. Similarly, circles $C2$, $C3$, $C4$ and $C5$ form the pairwise boundaries with $1.5+j$, $1.4+0.8j$, $1.8+0.6j$ and $2.5+0.5j$ respectively, for the singular fade state $2+j$ as shown in Fig. $\ref{fig:onesfs}$. The pairwise transition boundary between $2+j$ and $2+2j$ is the straight line $L1$ shown in Fig. $\ref{fig:onesfs}$. Further, from Fig. $\ref{fig:onesfs}$, lines $L2$ and $L3$ are the pairwise transition boundaries with $3+j$ and $2$ respectively for the singular fade state $2+j$. Finally, the region ${\cal R}_{\left\{ 2+j\right\}}$ is the shaded region shown in Fig. $\ref{fig:2_1jsfs}$.
%%%%%%%%%%%%%%%%%%%%%%%%%%%%%%%%%%%%%%%%%%%5
\begin{figure}[h]
\centering
\subfigure[Pairwise transition boundaries corresponding to $2+j$]{
\includegraphics[totalheight=2.5in,width=3.1in]{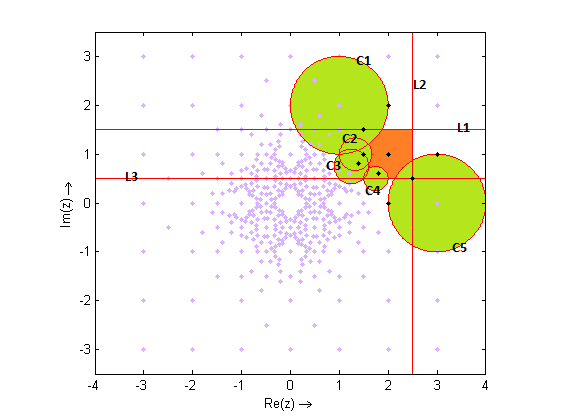}
\label{fig:onesfs}
}
%\vspace{-0.5 in}
\subfigure[Region ${\cal R}_{\left\{ 2+j\right\}}$ for 16-QAM]{
\includegraphics[totalheight=2.5in,width=3 in]{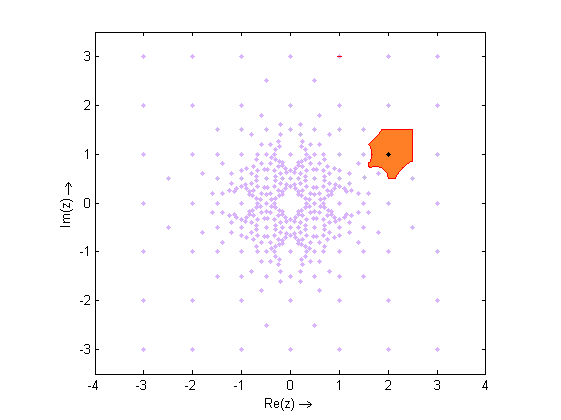}
\label{fig:2_1jsfs}
}
\label{example_2_j}
\caption{Diagram explaining the procedure to get region for a singular fade state}
\vspace{-0.2 in}
\end{figure}

\begin{figure}
\centering
\includegraphics[totalheight=2.5in,width=2.75in]{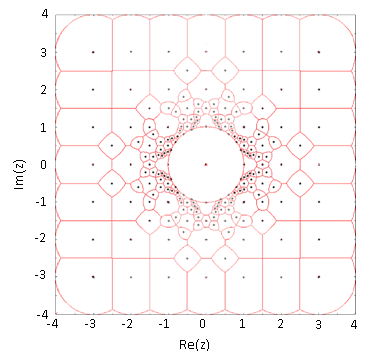} 
\caption{Diagram showing the regions associated with the singular fade states lying outside the unit circle for 16-QAM signal set}
\label{fig:full1}
\end{figure}

\begin{figure}[h]
\vspace{-0.2 in}
%\centering
\includegraphics[totalheight=3in,width=3.75in]{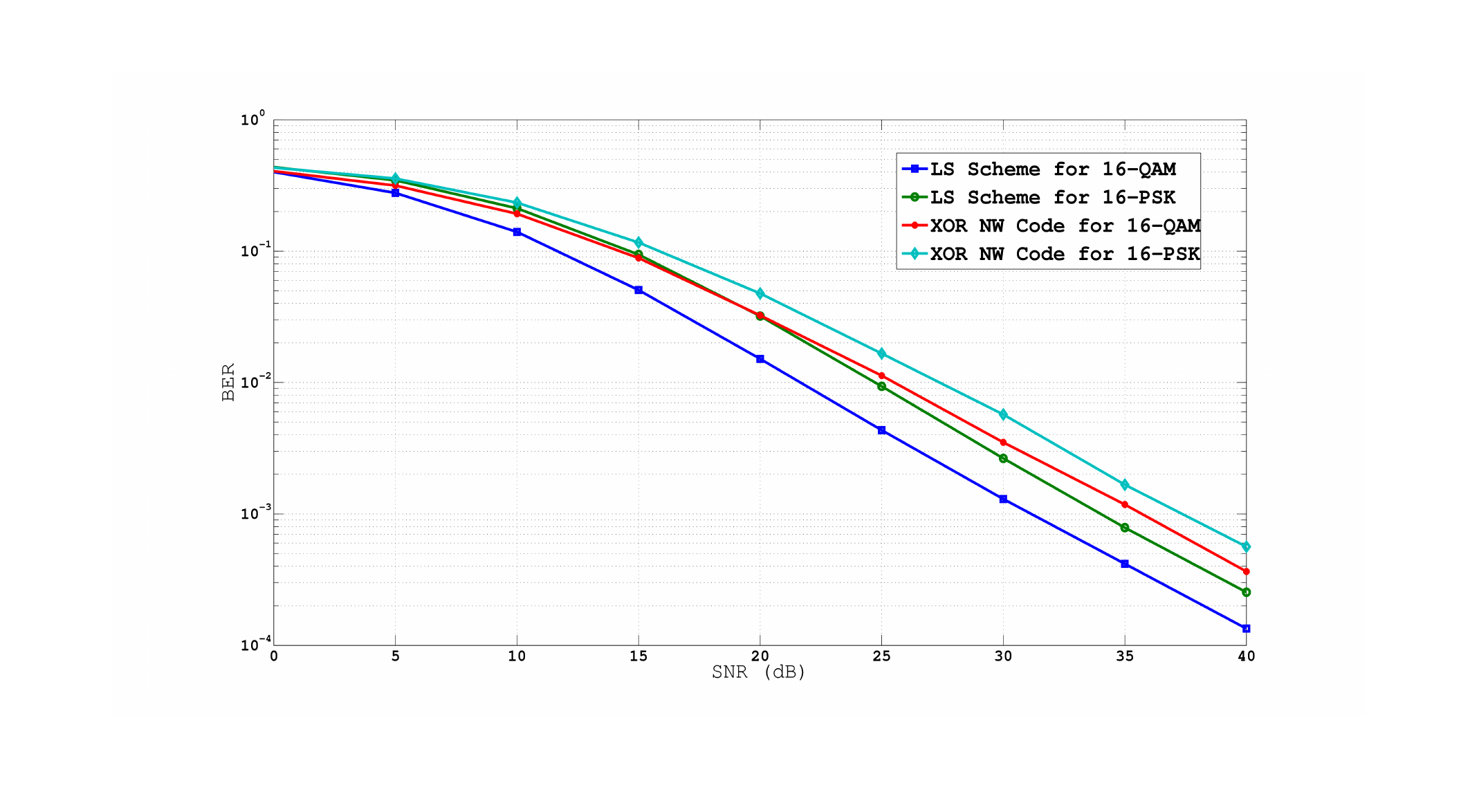}
\vspace{-0.6 in}
\caption{SNR vs BER for different schemes when the end nodes use 16-QAM and 16-PSK for a Rayleigh fading scenario.}     
\label{fig:rayleigh}        
\end{figure}

\begin{figure}[h]
\centering
\vspace{-0.35 in}
\includegraphics[totalheight=3in,width=3.75in]{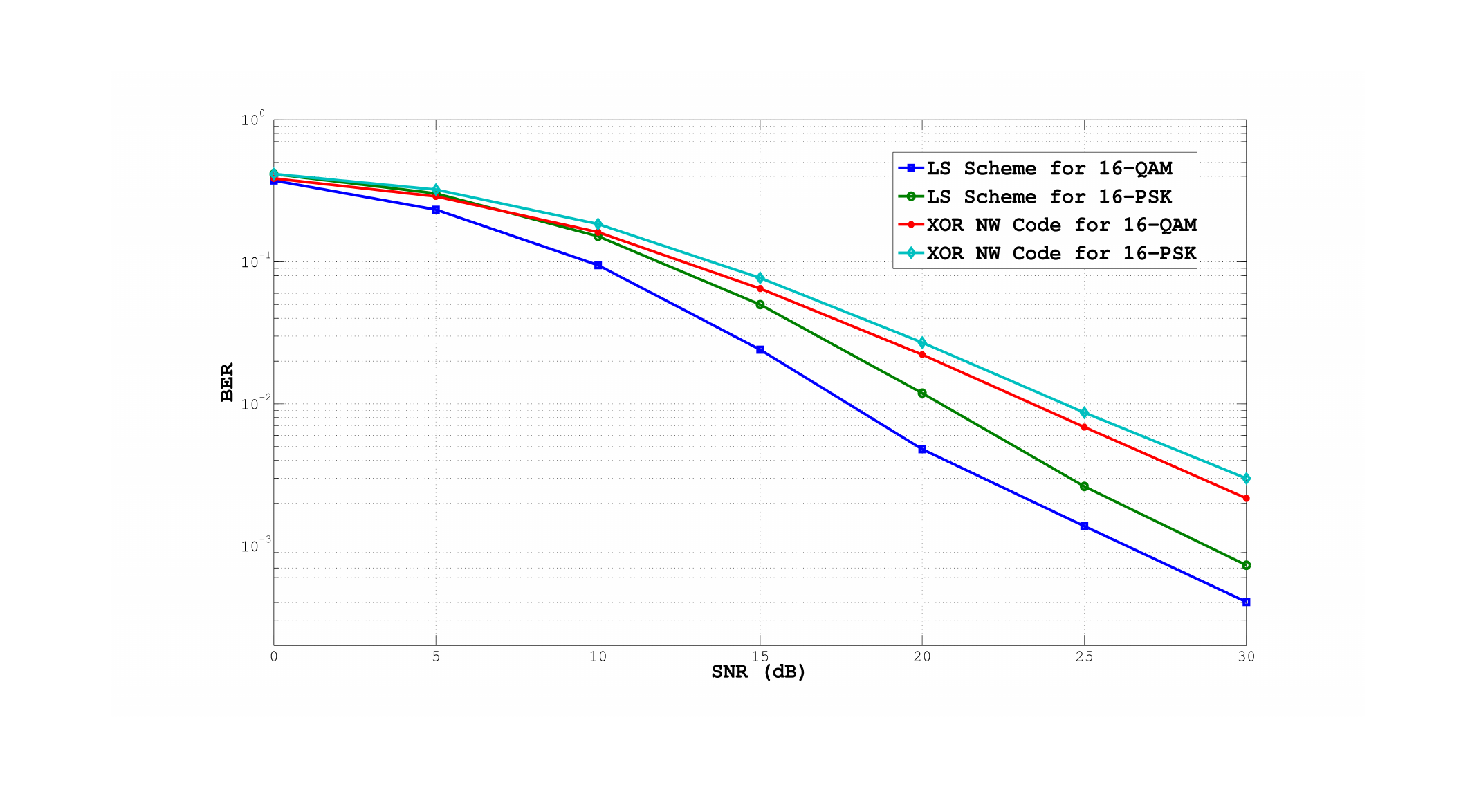}
\vspace{-0.6 in}
\caption{SNR vs BER for different schemes when the end nodes use 16-QAM and 16-PSK for a Rician fading scenario with Rician factor 5 dB. }     
\label{fig:rician}        
\end{figure}

Proceeding likewise for all the singular fade states lying outside the unit circle, we get the channel quantization for 16-QAM signal set as shown in Fig. $\ref{fig:full1}$.

%%%%%%%%%%%%%%%%%%%%%%%%%%%%%%%%%%%%%%%%%%%%%%%%%%%%%%%%%%%%%%%%%%%%%%%%%%%%%%%%%%%%%%%%%%%%
\section{Simulation Results}
\label{sec5}
%%%%%%%%%%%%%%%%%%%%%%%%%%%%%%%%%%%%%%%%%%%%%%%%%%%%%%%%%%%%%%%%%%%%%%%%%%%%%%%%%%

%%%%%%%%%%%%%%%%%%%%%%%%%%%%%%%%%%%%%%%%%%%%%%%%%%%%%%%%%%%%%%%%%%%%%%%%%%%%%%%%%%%%%%%%%%
The Latin Square (LS) scheme \cite{NVR} is based on removing the singular fade states. For 16-PSK all the 912 singular fade states can be removed with Latin Squares with number of symbols 16, but for 16-QAM some singular fade states cannot be removed with Latin Squares of 16 symbols, we used Latin Squares consisting 20 symbols for some singular fade states. Since 16-QAM has only 388 singular fade states, in comparison with 912 singular fade states of 16-PSK, and since 16-QAM offers better distance distribution in the MA stage $16$-QAM gives better performance. For a given average energy, the end to end BER is a function of distance distribution of the constellations used at the end nodes as well as at the relay. The simulation results for the end to end BER  as a function of SNR is presented for different fading scenarios.

Consider the case when $h_A , h_B , h_{A}^{\prime}$ and $h_{B}^{\prime}$ are distributed according to Rayleigh distribution, with the variances of all the fading links are assumed to be 0 dB. The end to end BER as a function of SNR in dB when the end nodes use 16-QAM signal sets as well as 16-PSK signal sets with same average energy is given in Fig.\ref{fig:rayleigh}. The end to end BER for XOR network code for 16-QAM is also given. It can be observed that the LS Scheme for 16-QAM outperforms LS Scheme for 16-PSK as well as XOR network code. 

Consider the case when $h_A , h_B , h_{A}^{\prime}$ and $h_B^{\prime}$ are distributed according to Rician distribution, with the Rician factor of 5 dB and the variances of all the fading links are assumed to be 0 dB. In Fig.\ref{fig:rician} the end to end BER as a function of SNR in dB for LS scheme for 16-PSK, 16-QAM and XOR network coding for 16-QAM is given. It is observed that the LS scheme gives large gain over the XOR network coding scheme. The LS scheme for QAM is  better in end to end BER performance in comparison with the LS scheme for PSK.  
%%%%%%%%%%%%%%%%%%%%%%%%%%%%%%%%%%%%%%%%%%%%%%%%%%%%%%%%%%%%%%%%%%%%%%%%%%%%%%%%%%%%%%%%%%%%
\section{Conclusion}
In this work, the design of modulation schemes for the physical layer network-coded two way relaying scenario when the end nodes use square QAM constellation is studied. We show that there are many advantages of using square QAM constellation over PSK signal set. Construction of the standard Latin square for removing the singular fade state $z$=1 for $M$-QAM is described and this is shown to be different from that used for $M$-PSK. Using the relation between exclusive law satisfying clusterings and Latin Squares, a method to remove all the other singular fade states is proposed and a means to derive the corresponding Latin squares is presented. This gives us all the maps to be used at the relay when square QAM constellation is used at the end nodes. The channel partition for QAM signal set is obtained analytically. Simulation results showing the end to end BER performance when the end nodes use PSK constellation as well as QAM constellations are obtained to support our claim.
%%%%%%%%%%%%%%%%%%%%%%%%%%%%%%%%%%%%%%%%%%%%%%%%%%%%%%%%%%%%%%%%%%%%%%%%%%%%%%%%%%%%%%%%%%%%%%%%%%%%%%%%%%%%%%%%%55  

%%%%%%%%%%%%%%%%%%%%%%%%%%%%%%%%%%%%%%%%%%%%%%%%%%%%%%%%%%%%%%%%%%%%%%%%%%%%%%%%%%%%%%

%%%%%%%%%%%%%%%%%%%%%%%%%%%%%%%%%%%%%%%%%%%%%%%%%%%%%%%%%%%%%%%%%%%%%%%%%%%%%%%%%%%%%%%%%55
\appendices

%%%%%%%%%%%%%%%%%%%%%%%%%%%%%%%%%%%%%%%%%%%%%%%%%%%%%%%%%%%%%%%%%%%%%%%%%%%%%%%%%%%%%%%%%%%%%

\section{PROOF OF LEMMA \ref{no_sing_pam}}
\label{app2}
There are $2(\sqrt M-1)$ non-zero signal points in the difference constellation $\Delta\mathcal{S}$ and since $ \Delta\mathcal{S}$ is symmetric about zero there are $\sqrt M-1$ signal points in $\Delta{S}^+$. All these are scaled version of nonzero elements of $\mathbb{Z}_{\sqrt M}$. %The number of singular fade states becomes a function of the number of different relatively prime pairs available in $\mathbb{Z}_{\sqrt M}$. 

The number of positive integers less than or equal to $n$ that are relatively prime to $n$ is given by, 
\begin{equation}
\label{app_euler}
 \psi(n)= n \prod_{p|n} \left(1-\frac{1}{p}\right)
\end{equation}
\noindent
where the product is taken over distinct prime numbers $p$ dividing $n$.  To get the total number of  relatively prime pairs in $\mathbb{Z}_{\sqrt M}$, we take the sum over all nonzero $n \in \mathbb{Z}_{\sqrt M}$  which gives $ \sum_{n=1}^{\sqrt M -1}n \prod_{p\vert n} \left(1-\frac{1}{p}\right).$ One relatively prime pair $(a,b)$ gives two singular fade states, $a/b$ and $b/a$. The multiplication factor $4$ in \eqref{sum_euler} accounts for the  negative side of the in-phase axis as well as the inverses. Finally, 2 is added to count the singular fade state $z=1$ and $z=-1$. \begin{flushright}
\vspace{-0.68 cm}
$\blacksquare$
\end{flushright}

\section{PROOF OF LEMMA \ref{refl_sym}}
\label{app7}
Let  $\gamma e^{j \theta}$ = $\gamma_{I}+j \gamma_{Q}$, such that $\gamma_{I}^2 + \gamma_{Q}^2 = \gamma^2$ and $tan^{-1}(\frac{\gamma_{Q}}{\gamma_{I}})=\theta$. Then $\gamma e^{j (90-\theta)}$ = $\gamma_{Q}+j \gamma_{I}$. We have,
\begin{align*}
\gamma e^{j \theta}= \dfrac{[x_A-x_A^{\prime}]}{[x_B^{\prime}-x_B]}.
\end{align*}

\noindent
Let $x_A-x_A^{\prime}= x_{d_{A}I}+j x_{d_{A}Q}$ and $x_B^{\prime}-x_B= x_{d_{B}I}+j x_{d_{B}Q}$. Then,
\begin{align*}
\gamma_{I}+j \gamma_{Q}= \dfrac{x_{d_{A}I}+j x_{d_{A}Q}}{x_{d_{B}I}+j x_{d_{B}Q}}.
 \end{align*}

\noindent
After expansion, we get
\begin{align}
\label{before_swap}
\nonumber
\gamma_{I}+j \gamma_{Q}&=\dfrac{(x_{d_{A}I} x_{d_{B}I} + x_{d_{A}Q} x_{d_{B}Q})}{(x_{d_{B}I}^{2}+x_{d_{B}Q}^{2})} +  \\ 
&~~~~~~~~~~~~~~~~~~\frac{j (x_{d_{A}Q} x_{d_{B}I}- x_{d_{A}I}x_{d_{B}Q})}{(x_{d_{B}I}^{2}+x_{d_{B}Q}^{2})}
\end{align}
%\noindent
and
\begin{align}
\label{after_swap}
\nonumber
\gamma_{Q}+j \gamma_{I}&=\dfrac{(x_{d_{A}Q} x_{d_{B}I} - x_{d_{A}I} x_{d_{B}Q})}{(x_{d_{B}I}^{2}+x_{d_{B}Q}^{2})} + 
\\
&~~~~~~~~~~~~~~~~~~\frac{j (x_{d_{A}I} x_{d_{B}I} + x_{d_{A}Q} x_{d_{B}Q})}{(x_{d_{B}I}^{2}+x_{d_{B}Q}^{2})}.
\end{align}

%\noindent
By comparing above expressions we can say that $x_{d_{A}I}$, $x_{d_{A}Q}$ ,$x_{d_{B}I}$ and $x_{d_{B}Q}$ in \eqref{before_swap} are changed to $x_{d_{A}Q}$, $x_{d_{A}I}$, $x_{d_{B}I}$ and $-x_{d_{B}Q}$ in \eqref{after_swap}. i.e., the difference constellation points $x_{d_{A}I}+j x_{d_{A}Q}$ and $x_{d_{B}I}+j x_{d_{B}Q}$ whose ratio gives singular fade state ($\gamma$, $\theta$) are converted to $x_{d_{A}Q}+j x_{d_{A}I}$ and $x_{d_{B}I}-j x_{d_{B}Q}$ and whose ratio gives singular fade state ($\gamma$, $(90-\theta)$). Let $x_k=x_{{k}I}+j x_{{k}Q}$, where $k \in \{ A,B,A^{\prime}, B^{\prime}\}$ and let the singularity removal constraint for $(\gamma, \theta)$,be $(x_A,x_B)(x_A^{\prime},x_B^{\prime})= (x_{{A}I}+j x_{{A}Q},x_{{B}I}+j x_{{B}Q})(x_{{A}I}^{\prime}+j x_{{A}Q}^{\prime},x_{{B}I}^{\prime}+j x_{{B}Q}^{\prime})$. The singularity removal constraint for $(\gamma, (90-\theta))$, can be written as $(x_{{A}Q}+j x_{{A}I},x_{{B}I}-j x_{{B}Q})(x_{{A}Q}^{\prime}+j x_{{A}I}^{\prime},x_{{B}I}^{\prime}-j x_{{B}Q}^{\prime})$. Since constellation points used by node A are indexed by rows in the Latin Square, the rows corresponding to constellation point $x_{{A}I}+j x_{{A}Q}$ are permuted to the rows denoted by $x_{{A}Q}+j x_{{A}I}$ and since the signal points used by node B are indexed by columns in the Latin Square, the columns corresponding to constellation point $x_{{B}I}+j x_{{B}Q}$ are permuted to the columns denoted by $x_{{B}I}-j x_{{B}Q}$, to get the Latin Square to remove $(\gamma,(90-\theta))$.
\begin{flushright}
\vspace{-0.68 cm}
$\blacksquare$
\end{flushright}

\section{PROOF OF LEMMA \ref{pam_std_lemma}}
\label{app8}
Consider the $\sqrt M$-PAM signal set with the signal points labelled from left to right as discussed in Section \ref{sec2}. Let $\{(k_1,l_1)(k_2,l_2)\}$ be a singularity removal constraint. To get the same point in the received constellation at the relay R, when $z=1,$ we have $k_1+l_1=k_2+l_2.$ Consider the following two cases satisfying this equality.
\noindent
Case (i): $k_2=l_1, l_2=k_1$ In this case the constraint becomes $\{(k_1,l_1)(l_1,k_1)\},$ i.e., the Latin Square which removes $z=1$ should be symmetric about main diagonal.\\
\noindent
Case (ii): $k_2=k_1+m, l_2=l_1-m$ for any $m \leq \sqrt M.$, The constraint now becomes $\{(k_1,l_1)(k_1+m,l_1-m)\}$ which means the symbol in $k_1$-th row and $l_1$-th column should be repeated in the $k_1+1$-th row and the $l_1-1$-th column. \\
It is easily seen that a left-cyclic Latin Square satisfies both this requirements.
\begin{flushright}
\vspace{-0.68 cm}
$\blacksquare$
\end{flushright}

\section{PROOF OF LEMMA \ref{qam_ls}}
\label{app9}
Note that the matrix in Fig. \ref{fig:qam_const} is a $M \times M$  matrix, which is also a $\sqrt M \times \sqrt M$ block left-cyclic  matrix where each block is a $\sqrt M \times \sqrt M$  left-cyclic matrix $L_{PAM(i)}$ for some $i.$      

Let $a_{1}+jb_{1}, a_{2}+jb_{2},a_{1}^{\prime}+jb_{1}^{\prime}$ and $a_{2}^{\prime}+jb_{2}^{\prime}$, where $a_{i},a_{i}^{\prime},b_{i}$ and $b_{i}^{\prime} \in \{-(\sqrt M-1),-(\sqrt M-3),\cdots,(\sqrt M-3),(\sqrt M-1)\}$ for $i \in \{1,2\}$ be four $M$-QAM constellation points such that $a_{1}+jb_{1}$ and $a_{1}^{\prime}+jb_{1}^{\prime}$ are used by node A and $a_{2}+jb_{2}$ and $a_{2}^{\prime}+jb_{2}^{\prime}$ are used by end node B, and result in a same point in the effective received constellation at the relay node for singular fade state $z=1$, i.e.,
\begin{align*}
a_{1}+jb_{1}+a_{2}+jb_{2}=a_{1}^{\prime}+jb_{1}^{\prime}+a_{2}^{\prime}+jb_{2}^{\prime}.
\end{align*}

Let $a_{1}^{\prime}=a_1+m_1$ and $b_{1}^{\prime}=b_1+m_2$ where $m_1,m_2 \in \{-2(\sqrt M -1), -2(\sqrt M -2),\cdots,2(\sqrt M -2),2(\sqrt M -1)\}$. Then, $a_{2}^{\prime}=a_2-m_1$ and $b_{2}^{\prime}=b_2-m_2.$ Then, using the map defined in \eqref{mumap}, let
%It can be seen that for a $a_1+jb_1$ and $a_2+jb_2$, we need to consider values of $m_1$ and $m_2$ only when $a_1+m_1+j(b_1+m_2)$ and $a_2-m_1+j(b_2-m_2)$ are $M$-QAM constellation points, such $m_1, m_2$ values are called valid. Let $\mu$ be tdenote the complex number corresponding to a M-QAM constellation point to symbol(from $\mathbb{Z}_M$) mapping. Let
\begin{align*}
k_1&=\mu(a_{1}+jb_{1}),\\ 
l_1&=\mu(a_{2}+jb_{2}),\\
k_2=\mu(a_{1}^{\prime}+jb_{1}^{\prime})&=\mu(a_1+m_1+j(b_1+m_2))~~\mathrm{and}\\
l_2=\mu(a_{2}^{\prime}+jb_{2}^{\prime})&=\mu(a_2-m_1+j(b_2-m_2)).
\end{align*}
Since, for $z=1,$ the four complex numbers result in the same point in the effective constellation at the relay, $\{(k_1,l_1)(k_2,l_2)\}$ is a singularity removal constraint for $z=1.$ From the 
above equations it follows that
%mapping $\mu$ defined in section \ref{sfs_of_qam} $\mu(A_{mI}+jA_{mQ})=\frac{1}{2}[(\sqrt M -1 +A_{mI})\sqrt M + (\sqrt M -1 +A_{mQ})]$, it is obtained that for valid $m_1$ and $m_2$,
\begin{align*}
k_2=k_1+\frac{1}{2} (m_1 \sqrt M+m_2)\\
l_2=l_1-\frac{1}{2}(m_1 \sqrt M+m_2),
\end{align*}
The above equations precisely mean the construction shown in Fig.\ref{fig:qam_const}. This completes the proof. 
\begin{flushright}
\vspace{-0.25 in}
$\blacksquare$
\end{flushright}

\section{PROOF OF THEOREM $\ref{Thrm1}$ }
\label{appen1}
From ($\ref{eqnGCI}$) and ($\ref{eqnGCIext}$) we have
\begin{align}
\nonumber
\label{appen1eqn1}
\Gamma_{CI}^{ext}({\cal S}) = \{ \gamma e^{j\theta} : \left| d_k + \gamma e^{j\theta} d_l \right| \geq \min	( d_{min}({\cal S}), \gamma d_{min}({\cal S}) ) \\
\forall ( d_k, d_l )\neq ( 0, 0 ) \in \Delta {\cal S} \times \Delta {\cal S}, \gamma > 1, - \pi \leq \theta < \pi \}.
\end{align}
Let $-d_k/d_l = z \in {\cal H}$ with $|z|>1$. Then, $\forall \gamma e^{j\theta} \in \Gamma_{CI}^{ext}\left({\cal S}\right)$, we have
\begin{align}
\label{appen1eqn2}
\nonumber
&|d_l|^2 |-z+\gamma e^{j\theta}|^2 \geq \left(d_{min}\left( {\cal S}\right)\right)^2  \\ &~~~~~~~~~~~~~\forall d_l \neq 0 \in {\Delta\cal S} , \forall z=-d_k/d_l \in {\cal H} ~\mathrm{with}  ~|z|>1.
\end{align}
In particular, $\forall \gamma e^{j\theta} \in \Gamma^{ext}_{CI}\left({\cal S}\right)$
\begin{align}
\label{eqnextineq}
\min_{d_l \neq 0 \in {\Delta\cal S}} \left \{ |d_l|^2 |-z+\gamma e^{j\theta}|^2 \right \}\geq \left(d_{min}\left( {\cal S}\right)\right)^2.
\end{align}
Now $\min_{d_l \neq 0 \in {\Delta\cal S}}\left \{ |d_l|^2 \right \}= \left(d_{min}\left( {\cal S}\right)\right)^2= 4/\rho$, from ($\ref{eqndmin}$) for $M$-QAM. Using ($\ref{eqndifc}$), $|d_l|=2/\sqrt{\rho}$ implies that $d_l$ is of the form $\pm 2/\sqrt{\rho}$ or $ \pm j  2/\sqrt{\rho}$. For such $d_l$, $z=-d_k/d_l$ is of the form $\alpha + j\beta$, where $\alpha, \beta \in {\cal S''} \triangleq {\left\lbrace {-\left(\sqrt{M}-1 \right)}, {-\left(\sqrt{M}-2 \right)},..., {\left(\sqrt{M}-2 \right)}, {\left(\sqrt{M}-1 \right)} \right\rbrace.}$
So ($\ref{eqnextineq}$) reduces to the form
\begin{align}
\label{eqnoutercircles}
|-\alpha -j\beta + \gamma e^{j\theta}|^2\geq 1,~\forall \alpha, \beta \in {\cal S''}.
\end{align}
This is the exterior of the unit circle with center ($\alpha,\beta$). It must be noted here that $\forall \gamma e^{j\theta} \in \Gamma^{ext}_{CI}({\cal S})$ and $\forall \alpha, \beta \in {\cal S''}$, ($\ref{eqnoutercircles}$) must hold.
If ${(\alpha,\beta) \in \left({\cal S''}-\{\pm (\sqrt{M}-1 )\} \right) \times \left({\cal S''} -\{\pm (\sqrt{M}-1 )\} \right)}$, ($\ref{eqnoutercircles}$) does not hold for ($\alpha,\beta$) $\in \left\lbrace\pm \left(\sqrt{M}-1 \right)\right\rbrace \times \left\lbrace\pm \left(\sqrt{M}-1 \right)\right\rbrace $.
So either one or both of $\alpha$ and $\beta$ must belong to the set $\left\lbrace\pm \left(\sqrt{M}-1 \right)\right\rbrace $ and for such a choice, ($\ref{eqnoutercircles}$) holds for the outer envelop region of the unit circles with centers ($\alpha,\beta$). This means
\begin{align}
\label{eqncentres}
\nonumber
\alpha+j\beta \in &\left\lbrace \pm\left(\sqrt{M}-1 \right)+ jx, x \pm j\left(\sqrt{M}-1 \right)\right\rbrace,\\ &~~~~~~~~~~~~~~~~~~~~~~~~~~~~~~~~ \mathrm{where} ~x \in {\cal S''}.
\end{align}
From ($\ref{eqncentres}$) the number of such circles is $2\times 2\left( 2\sqrt{M}-1\right)-4 = 8\left( \sqrt{M}-1\right)$.
\begin{flushright}
\vspace{-0.25 in}
$\blacksquare$
\end{flushright}

\section{PROOF OF THEOREM $\ref{Thrm2}$ }
\label{appen3}
At the pair wise transition boundary corresponding to the pair of clusterings ${\cal C}^{\left\lbrace -d_{k_1}/d_{l_1} \right\rbrace}$ and ${\cal C}^{\left\lbrace -d_{k_2}/d_{l_2} \right\rbrace}$, denoting $-d_{k_1}/d_{l_1}$ as $z_1$ and $-d_{k_2}/d_{l_2}$ as $z_2$, the set of values of $\gamma e^{j\theta}$ satisfies
\begin{align}
\label{appn2eqn1}
\min_{\substack{{d_{l_1}:} \\ {\frac{-d_{k_1}}{d_{l_1}}=z_1}}} \left\lbrace \vert d_{k_1} + \gamma e^{j\theta} d_{l_1} \vert \right\rbrace = \min_{\substack{{d_{l_2}:} \\ {\frac{-d_{k_2}}{d_{l_2}}=z_2 \neq z_1}}} \left\lbrace \vert d_{k_2} + \gamma e^{j\theta} d_{l_2} \vert \right\rbrace.
\end{align}

Say the $min$ in the LHS and RHS of ($\ref{appn2eqn1}$) are achieved at $\check{d_l}$ and $\check{d_{l'}}$ respectively. Using the notation used in ($\ref{eqndifc}$), ($\ref{appn2eqn1}$) becomes,
%\begin{scriptsize}
\begin{align}
\label{appen2eqn2}
\nonumber
\vert n_{k_1} +jm_{k_1} + \gamma e^{j\theta} &\left(\check{n_{l}} +j\check{m_{l}} \right) \vert = \\
&\vert n_{k_2} +jm_{k_2} + \gamma e^{j\theta} \left( \check{n_{l'}} +j\check{m_{l'}} \right) \vert.
\end{align}
%\end{scriptsize}
With $\gamma e^{j\theta} = \gamma_R + j\gamma_I$, squaring and rearranging ($\ref{appen2eqn2}$) gives
\begin{align}
\nonumber
\label{appen2eqn3}
&{( {\check{n_l}^2+\check{m_l}^2-{\check{n_{l'}}}^2-{\check{m_{l'}}}^2})\gamma_R^2+({\check{ n_l}^2+\check{m_l}^2-{\check{n_{l'}}}^2-{\check{m_{l'}}}^2})\gamma_I^2}
\nonumber
\\ &~~~~~+ 2\left( n_{k_1}\check{n_l}+m_{k_1}\check{m_l}-n_{k_2}\check{n_{l'}}- m_{k_2}\check{m_{l'}}\right)\gamma_R  
\nonumber
\\ &~~~~~~~~~~+ 2\left( m_{k_1}\check{n_l}-n_{k_1}\check{m_l}-m_{k_2}\check{n_{l'}}+n_{k_2}\check{m_{l'}}\right)\gamma_I 
\nonumber
\\&~~~~~~~~~~~~~~~~~~= {n_{k_2}}^2+{m_{k_2}}^2-{n_{k_1}}^2-{m_{k_1}}^2
\end{align}
Now, since 
\begin{align}
\nonumber
\label{appen2eqn4}
\Re \left( -d_k/d_l\right) &= \Re \left( -\frac{n_k+jm_k}{n_l+jm_l}\right)\\
&=-\frac{n_{k}n_l+m_{k}m_l}{n_l^2+m_l^2}\\
\nonumber
\mathrm{and}\\
\nonumber
\Im \left( -d_k/d_l\right) &= \Im \left( -\frac{n_k+jm_k}{n_l+jm_l}\right)\\
&=-\frac{m_{k}n_l-n_{k}m_l}{n_l^2+m_l^2}
\end{align}
($\ref{appen2eqn3}$) becomes
\begin{align}
\nonumber
\label{appen2eqn5}
&\left( |\check{d_l}|^2- |\check{d_{l'}}|^2\right)\left( {\gamma_R}^2+{\gamma_I}^2\right)-2\left( \Re\left( z_1\right)|\check{d_l}|^2- \Re\left( z_2\right)|\check{d_{l'}}|^2\right)\\ &~~~-2\left( \Im\left( z_1\right)|\check{d_l}|^2- \Im\left( z_2\right)|\check{d_{l'}}|^2\right)
=|z_2|^2 |\check{d_{l'}}|^2-|z_1|^2 |\check{d_{l}}|^2
\end{align}
If $|\check{d_l}|^2 \neq |\check{d_{l'}}|^2$, ($\ref{appen2eqn5}$) is the equation of a circle of the form 
\begin{align*}
{\left( \gamma_R-x\right)}^2 + {\left( \gamma_I-y\right)}^2 = r^2
\end{align*}
with $x$, $y$ and $r$ as given in Theorem $\ref{Thrm2}$. \\
If $|\check{d_l}|^2 = |\check{d_{l'}}|^2$, ($\ref{appen2eqn5}$) is a linear equation. Since $z_1 \neq z_2$, the coefficients of both $\gamma_R$ and $\gamma_I$ cannot be simultaneously zero, so that we get the equation of a straight line.
\begin{flushright}
\vspace{-0.68 cm}
$\blacksquare$
\end{flushright}
%%%%%%%%%%%%%%%%%%%%%%%%%%%%%%%%%%%%%%%%%%%%%%%%%%%

%%%%%%%%%%%%%%%%%%%%%%%%%%%%%%%%%%%%%%%%%%%%%%%%%%
\end{document}